\documentclass[reqno,12pt,a4paper]{article}

\usepackage{amsfonts,amssymb,amsthm, 
amsmath,amscd, bbm, bm, mathabx, 
mathrsfs,delarray,subfigure}
\usepackage[dvipsnames]{xcolor}
\usepackage{hyperref}
\usepackage{cite}
\usepackage{xypic}
\usepackage{sectsty}

\usepackage{accents}

\sectionfont{\fontsize{12}{15}\selectfont}
\subsectionfont{\fontsize{12}{15}\selectfont}

\usepackage{enumitem}

\makeatletter 
\newcommand\raggedtop{%
  \topskip=1\topskip plus 10pt} 

\makeatother 
\raggedtop

\DeclareMathOperator{\ad}{ad}
\DeclareMathOperator{\Ad}{Ad}

\DeclareMathOperator{\ran}{ran} 
\DeclareMathOperator{\id}{id}

\DeclareMathOperator{\Hom}{Hom}

\DeclareMathOperator{\Fun}{Fun}
\DeclareMathOperator{\iFun}{\textsc{Fun}} 

\DeclareMathOperator{\Map}{Map}
\DeclareMathOperator{\iMap}{\textsc{Map}}        

\DeclareMathOperator{\Aut}{Aut}

\DeclareMathOperator{\Der}{Der}

\DeclareMathOperator{\Lie}{Lie}

\DeclareMathOperator{\ee}{e}

\DeclareMathOperator{\ZZ}{Z\hspace{-1pt}}
\DeclareMathOperator{\DD}{D\hspace{-1pt}}

\DeclareMathOperator{\OOO}{Op\hspace{-1pt}}
\DeclareMathOperator{\iOOO}{\textsc{Op}\hspace{-1pt}}     

\numberwithin{equation}{subsection} 
\numberwithin{subsection}{section} 

\newcommand{\ceqref}[1]{{\textcolor{blue}{\eqref{#1}}}}
\newcommand{\cref}[1]{{\textcolor{blue}{\ref{#1}}}}
\newcommand{\ccite}[1]{{\textcolor{blue}{\!\cite{#1}}}}

\newcommand{\ddd}{{\hbox{$\bigoplus$}}}

\newcommand{\ul}[1]{{\underline{#1}}}
\newcommand{\bfs}[1]{{\boldsymbol{#1}}}

\newcommand{\mathsans}[1]{{{\sf #1}}}

\font\euler=eusm10 at 12.8 truept
\font\scripteuler=eusm7
\font\scriptscripteuler=eusm5 
\def\eul{\fam=12}
\newcommand{\matheul}[1]{{{\eul #1}}}

\newtheorem{defi}{{\sf Definition}}[section]
\newtheorem{prop}{{\sf Proposition}}[section]
\newtheorem{theor}{{\sf Theorem}}[section]
\newtheorem{lemma}{{\sf Lemma}}[section]

\begin{document}

\vskip1.5cm
\begin{large}
{\flushleft\textcolor{blue}{\sffamily\bfseries Operational total space theory of principal 2--bundles I:}}  
{\flushleft\textcolor{blue}{\sffamily\bfseries operational geometric framework}}  
\end{large}
\vskip1.3cm
\hrule height 1.5pt
\vskip1.3cm
{\flushleft{\sffamily \bfseries Roberto Zucchini}\\
\it Dipartimento di Fisica ed Astronomia,\\
Universit\`a di Bologna,\\
I.N.F.N., sezione di Bologna,\\
viale Berti Pichat, 6/2\\
Bologna, Italy\\
Email: \textcolor{blue}{\tt \href{mailto:roberto.zucchini@unibo.it}{roberto.zucchini@unibo.it}}, 
\textcolor{blue}{\tt \href{mailto:zucchinir@bo.infn.it}{zucchinir@bo.infn.it}}}


\vskip.7cm
\vskip.6cm 
{\flushleft\sc
Abstract:} 
It is a classic result that the geometry of the total space of a principal bundle
with reference to the action of the bundle's structure group is codified in the 
bundle's operation, 
a collection of derivations comprising the de Rham differential and the contraction
and Lie derivatives of all vertical vector fields and obeying the six Cartan relations. 
In particular, connections and gauge transformations can be defined through the way they 
are acted upon by the operation's derivations. In this paper, the first of a series of two
extending the ordinary theory, we construct an operational total space theory of strict principal 2--bundles 
with regard to the action of the structure strict 2--group. Expressing this latter via a crossed module
$(\mathsans{E},\mathsans{G})$, the operation is based on the derived Lie group
$\mathfrak{e}[1]\rtimes\mathsans{G}$. In the second paper, 
an original formulation of the theory of $2$--connections and $1$-- and $2$--gauge transformations
based on the operational framework worked out here will be provided. 
\vspace{2mm}

\vfil\eject
\tableofcontents

\vfil\eject

\section{\textcolor{blue}{\sffamily Introduction}}\label{sec:intro}

\vspace{-.63mm}

Principal $2$--bundles with strict structure $2$--group are most often 
described through appropriate sets of transitional data with respect to an open cover of the base manifold, 
or cocycles, a framework  originally put forward in the foundational works by Schreiber \ccite{Schreiber:2005ff} 
and Baez and Schrieber \ccite{Baez:2004in,Baez:2005qu}. This approach points directly to Giraud’s 
non Abelian cohomology \ccite{Giraud:1971cna} which provides the associated classifying theory.
The first formulation of the total space theory of principal $2$--bundles as Lie groupoids 
was given by Bartels \ccite{Bartels:2006hgtb} and Baez and Schreiber \ccite{Baez:2004in,Baez:2005qu},
making systematic use of the ideas and techniques of categorification. The total space perspective was further developed 
by Wockel \ccite{Wockel:2008tspb}, who also obtained a classification of principal $2$--bundles up to 
Morita equivalence by non Abelian cohomology and provided a categorical characterization of the 
the gauge $2$--group. Schommer-Pries \ccite{Schommer:2011ces} extended Wockel's theory
by constructing a bicategory of principal $2$--bundles. 

Special incarnations of principal $2$--bundles are provided by the non Abelian bundle gerbes 
of Aschieri, Cantini and Jurco \ccite{Aschieri:2003mw} and the G--gerbes of Laurent--Gengoux, 
Sti\'enon and Xu \ccite{Laurent:2009nag} and Ginot and Sti\'enon \ccite{Ginot:2008gpc}. 
According to Nikolaus and Waldorf \ccite{Nikolaus:2013fvng}, there are indeed three more
distinct but equivalent ways of presenting by principal $2$--bundles with strict structure $2$--group beside 
the one already recalled: $i)$ by smooth non Abelian \v Cech cocycles of the $2$--group;
$ii)$ by classifying maps to classifying space of the $2$--group; $iii)$ by $2$--group 
bundle gerbes. \raggedbottom

The theory of $2$--connections and $1$-- and $2$--gauge transformations on principal $2$--bundles has been the 
object of intense inquiry because of its potential applications in higher gauge theory and string theory. 
It takes different forms 
depending on the description of $2$--bundles used. We mention in particular the work of Breen and Messing \ccite{Breen:2001ie}, 
Schreiber \ccite{Schreiber:2005ff,Schreiber:2013pra}, Baez and 
Schreiber \ccite{Baez:2004in,Baez:2005qu} and Jurco, Saemann and Wolf \ccite{Jurco:2014mva,Jurco:2016qwv} characterizing 
in various ways connections and gauge transformations in terms of transition data, 
Aschieri, Cantini and Jurco \ccite{Aschieri:2003mw} treating connections and gauge transformations 
in the framework of bundle gerbes and Laurent--Gengoux, Sti\'enon and Xu \ccite{Laurent:2009nag}
considering connections of G--gerbes. A theory of $2$--connections
working in Wockel’s total space principal $2$--bundle theory \ccite{Wockel:2008tspb} has been worked out only recently
by Waldorf \ccite{Waldorf:2016tsct,Waldorf:2017ptpb}.


\subsection{\textcolor{blue}{\sffamily The operational theory and its scope }}\label{subsec:scope}

The present endeavours puts forward a proposal for an operational total space theory of principal 
$2$--bundles. We begin by reviewing the operational approach to ordinary principal bundle theory
on which the higher extension is modelled. 

Let $P$ be a principal bundle $P$ over $M$ with structure group $\mathsans{G}$.  $P$ is so a manifold with a free 
fiberwise transitive right $\mathsans{G}$--action. With any Lie algebra element
$x\in\mathfrak{g}$ there is then associated a vertical vector field $C_{Px}$. 
It is a classic result \ccite{Greub:1973ccc} that the underlying total space geometry of $P$ 
is codified in its {\it operation} $\OOO S_P$. This is the geometrical structure consisting of the graded algebra 
$\Omega^\bullet(P)$ of differential forms of $P$ and the collection of graded derivations of
$\Omega^\bullet(P)$ comprising the de Rham differential $d_P$ and the contraction and Lie derivatives 
$j_{Px}$, $l_{Px}$ of the vector fields $C_{Px}$.   
The derivations obey the six Cartan relations,
\begin{align}
&[d_P,d_P]=0,
\vphantom{\Big]}
\label{intro1}
\\
&[d_P,j_{Px}]=l_{Px},
\vphantom{\Big]}
\label{intro2}
\\
&[d_P,l_{Px}]=0,
\vphantom{\Big]}
\label{intro3}
\\
&[j_{Px},j_{Py}]=0,
\vphantom{\Big]}
\label{intro4}
\\
&[l_{Px},j_{Py}]=j_{[x,y]},
\vphantom{\Big]}
\label{intro5}
\\
&[l_{Px},l_{Py}]=l_{[x,y]}.
\vphantom{\Big]}
\label{intro6}
\end{align} 
The differential forms of $P$ annihilated by all derivations $j_{Px}$, $l_{Px}$ with $x\in\mathfrak{g}$ 
are called {\it basic} and form a subalgebra $\Omega^\bullet{}_{\mathrm{b}}(P)$
of $\Omega^\bullet(P)$. $\Omega^\bullet{}_{\mathrm{b}}(P)$ is in one--to--one correspondence and thus can be identified with
the differential form algebra $\Omega^\bullet(M)$ on $M$.

A connections $\omega$ can be characterized as a $\mathfrak{g}$--valued $1$--form of $P$
acted upon in a prescribed way by the  operation's derivations.
Similarly, a gauge transformation $g$ can be defined 
as a $\mathsans{G}$--valued map behaving in a certain way under the action of the  
derivations. Gauge transformations form an infinite dimensional Lie group 
with a left action on the space of connections. 

On any sufficiently small neighborhood $U$ of $M$, on which the $\mathsans{G}$--bundle $P$ is isomorphic
to the trivial $\mathsans{G}$--bundle $U\times\mathsans{G}$, a connection $\omega$ is completely characterized by 
a basic $\mathfrak{g}$--valued $1$--form $\omega_{\mathrm{b}}$ on $U$. Similarly, 
a gauge transformation $g$ is characterized by a basic $\mathsans{G}$--valued map $g_{\mathrm{b}}$ on $U$.
Local connection and gauge transformation data 
relative to distinct overlapping trivializing neighborhoods of $M$ match in a prescribed way codified in
a $\mathsans{G}$--valued \v Cech $1$--cocycle. All these properties are described by the {\it basic} formulation 
of principal bundle geometry, which is the one implicitly used in most physical literature. 


The operational approach to principal bundle theory admits an elegant reformulation in the language of 
{\it graded differential geometry}.  
For a principal $\mathsans{G}$--bundle $P$, the graded differential form algebra
$\Omega^\bullet(P)$ is described as the graded function algebra $\Fun(T[1]P)$ of the shifted tangent bundle
$T[1]P$ of $P$ and the operation derivations $d_P$, $j_{Px}$, $l_{Px}$, $x\in\mathfrak{g}$, 
as graded vector fields on $T[1]P$. 
Connection and gauge transformations can then be defined as earlier as degree $1$ $\mathfrak{g}$--
and degree $0$ $\mathsans{G}$--valued 
functions on $T[1]P$
behaving in a prescribed way under the action of $d_P$, $j_{Px}$, $l_{Px}$. 
Gauge transformation of connections can be implemented in the familiar way. The basic formulation goes 
through essentially with no changes. 

Our investigation has shown that an analogous graded geometric operational total space formulation 
of principal 2--bundle theory can be worked out. 
The results of our research are expounded in a series of two papers.
In the first paper, referred to as I,
we set the foundations of the operational total space theory 
of principal 2--bundles. In the second paper, referred to as II \ccite{Zucchini:2019rps}, 
based on the operational setup worked out in I,
we provide an original formulation of the theory of $2$--connections and $1$-- 
and $2$--gauge transformations of strict principal 2--bundles. 

To make the nature of our work more easily appreciable
by the reader, we now outline briefly the intuitive ideas underlying the higher extension. 
The ordinary theory is the model against which the higher one is built. 
In a nutshell, the higher theory can be outlined as follows. With the
$2$--bundle's structure $2$--group, there is associated a derived graded group.
The structure $2$-group's action on the $2$--bundle is then shown to induce a
derived group's action on a synthetic form of the $2$--bundle. 
With the latter, there is associated an operation, that suitably describes the 
$2$--bundle's geometry and in particular $2$--connections and $1$-- 
and $2$--gauge transformations. 

The passage from the ordinary to the higher theory is not free of subtleties. 
The graded nature of the derived group makes it necessary to further reshape the ordinary theory 
before constructing the higher one by analogy. For a principal bundle $P$, 
the ordinary graded function algebra $\Fun(T[1]P)$ must be replaced 
by the {\it internal} graded function algebra $\iFun(T[1]P)$. While $\Fun(T[1]P)$ has a single grading, 
the shifted tangent bundle grading of $T[1]P$, $\iFun(T[1]P)$ has two, the tangent bundle 
grading and a further internal grading. 
$\Fun(T[1]P)$ can therefore be identified with the subalgebra of $\iFun(T[1]P)$ of zero internal grading. 

Extending the function algebra from $\Fun(T[1]P)$ to $\iFun(T[1]P)$ 
as indicated introduces in the original operational framework internal multiplicities.
In the familiar formulation of gauge theory, this would correspond to add the ghost degree
to the form one. It endows in this way connections and gauge transformations with ghostlike partners 
making the whole geometrical framework akin to that employed in 
the AKSZ formulation of BV theory \ccite{Alexandrov:1995kv} (see also \ccite{Zucchini:2017nax}).
In ordinary principal bundle geometry doing so is discretionary.
In higher principal bundle theory, as it turns out, it is unavoidable. 





\subsection{\textcolor{blue}{\sffamily Main features of the operational theory of principal 2--bundles}}\label{subsec:approach}

In this paper, I of the series, we build an operation describing the fibered
geometry of a given principal $2$--bundle. The attendant geometric framework  
will constitute the backdrop for the $2$--connection and $1$-- and $2$--gauge transformation
theory of the companion paper, II of the series. The formulation proposed, which is non trivial in
many respects, is outlined in this subsection. 


The operational framework (cf. subsect. \cref{subsec:opers}) 
is the paradigm on which the whole architecture of the present work rests. 
It is by design the most appropriate approach for the study of the total 
space differential geometry of ordinary principal bundles 
(cf. subsect. \cref{subsec:ordtotal}) and, 
once suitably adapted, has the potential of shedding light 
into that of principal 2--bundles as argued below. 

A principal $2$--bundle consists of a morphism manifold $\hat P$ with an object submanifold $\hat P_0$ 
constituting a groupoid, a base manifold $M$, compatible projection maps 
$\hat\pi:\hat P\rightarrow M$ and $\hat\pi_0:\hat P_0\rightarrow M$
forming a functor, a morphism group $\hat{\mathsans{K}}$ with an object subgroup $\hat{\mathsans{K}}_0$ organized 
as a strict Lie $2$-group and compatible right actions $\hat R:\hat P\times\hat{\mathsans{K}}\rightarrow\hat P$ and 
$\hat R_0:\hat P_0\times\hat{\mathsans{K}}_0\rightarrow\hat P_0$ building a functor and respecting $\hat\pi$ and $\hat\pi_0$.
The $2$--bundle enjoys further the property of locally trivializability, that is on any small enough neighborhood $U$ 
of $M$ the groupoid $(\hat P|_U,\hat P_0|_U)$ is equivariantly projection preservingly 
equivalent 
to the groupoid $(U\times\hat{\mathsans{K}},U\times\hat{\mathsans{K}}_0)$ with the obvious projection
and right action structures. 


By virtue of the right $\hat{\mathsans{K}}$-- and $\hat{\mathsans{K}}_0$--actions on the manifolds 
$\hat{P}$ and $\hat{P}_0$, there are two operations 
associated with any given principal $2$--bundle, which we could include in an operational 
framework. However, proceeding in this naive way would not allow us to make contact with the standard 
formulation of strict higher gauge theory. To achieve this end, we have to make a further step: 
shifting to what we call the synthetic formulation for its formal affinity to smooth infinitesimal analysis
of synthetic differential geometry. 

In the synthetic approach (cf. subsects. \cref{subsec:synth}--\cref{subsec:2prinbundop}), 
one adjoins to the given principal $2$--bundle a synthetic structure 
consisting of morphism and object manifolds $P$ and $P_0$, the base manifold $M$, 
projection maps $\pi$ and $\pi_0$, morphism and object groups $\mathsans{K}$ and $\mathsans{K}_0$ 
and right $\mathsans{K}$-- and $\mathsans{K}_0$-- actions $R$ and $R_0$ on $P$ and $P_0$. 
The synthetic setup is obtained from the original non synthetic one through the following formal 
construction. Describe the strict Lie $2$--group $(\hat{\mathsans{K}},\hat{\mathsans{K}}_0)$ by its associated 
Lie group crossed module $(\mathsans{E},\mathsans{G})$ so that 
$\hat{\mathsans{K}}=\mathsans{E}\rtimes\mathsans{G}$ and $\hat{\mathsans{K}}_0=\mathsans{G}$.
Then, $\mathsans{K}=\mathfrak{e}[1]\rtimes\mathsans{G}$ and $\mathsans{K}_0=\mathsans{G}$. 
Next, formally extend the $\hat{\mathsans{K}}$--action $\hat R$ to a $\mathsans{K}$--action. Then, $P$ is
the $\mathsans{K}$--action image of $\hat P_0$ and $P_0=\hat P_0$, $R$ is the restriction of $\hat R$ to $P$
and $R_0=\hat R$. Here, $\mathsans{K}$ and $P$ must be regarded as spaces of functions from $\mathbb{R}[-1]$ 
to $\mathsans{E}\rtimes\mathsans{G}$ and $P$, respectively. 
We remark that, although the synthetic structure shares many of the properties of the underlying principal $2$--bundle,
it is not one since neither pairs $(\mathsans{K},\mathsans{K}_0)$ and $(P,P_0)$ 
have a groupoid structure. 

There are two operations $\iOOO S_{P}$ and $\iOOO S_{P0}$ which codify the $\mathsans{K}$-- and 
$\mathsans{K}_0$--actions $R$ and $R_0$ of $P$ and $P_0$. By their synthetic nature, as it will be shown 
in II, they have the right properties for making the desired contact with strict higher gauge theory.
We describe them in greater detail next. 

The Lie group $\DD\mathsans{M}=\mathfrak{e}[1]\rtimes\mathsans{G}$ and its subgroup 
$\DD\mathsans{M}_0=\mathsans{G}$ encountered above 
as well as the Lie algebras $\DD\mathfrak{m}=\mathfrak{e}[1]\rtimes\mathfrak{g}$ 
and its subalgebra $\DD\mathfrak{m}_0=\mathfrak{g}$ are called derived and
play a fundamental role in the architecture of our construction. As partly anticipated, 
$\DD\mathsans{M}$ may be thought of as the Lie group of internal functions 
$\ee^{\bar\alpha L}a$, $\bar\alpha\in\mathbb{R}[-1]$, with $a\in\mathsans{G}$ and $L\in\mathfrak{e}[1]$.
Similarly, $\DD\mathfrak{m}$ may be thought of as the Lie algebra of internal functions 
$x+\bar\alpha X$, $\bar\alpha\in\mathbb{R}[-1]$, with $x\in\mathfrak{g}$ and $X\in\mathfrak{e}[1]$.
Note however that, in the spirit of the superfield formulation of supersymmetric field theories, here 
$\bar\alpha$ is simply a book--keeping parameter and not an extra variable injected in the theory and allowed 
to vary in its full range. 
 
The synthetic morphism operation $\iOOO S_{P}$ consists of the internal function algebra 
$\iFun(T[1]P)$ acted upon the de Rham vector field $d_P$ and the contraction and Lie 
vector fields $j_{PZ}$, $l_{PZ}$, $Z\in\DD\mathfrak{m}$. In this case, it is necessary and not just
merely optional to use the internal function algebra rather than the ordinary one, because the 
latter would not be preserved by the $j_{PZ}$, $l_{PZ}$ due to the non trivial gradation of 
$\DD\mathfrak{m}$. 

The synthetic object operation $\iOOO S_{P0}$ has a similar constitution. 
The underlying function algebra is the internal algebra $\iFun(T[1]P_0)$
and the operation vector fields are the de Rham vector field $d_{P_0}$ and the contraction and Lie 
vector fields $j_{P_0Z_0}$, $l_{P_0Z_0}$, $Z_0\in\DD\mathfrak{m}_0$.  
In this case, due to the trivial gradation of $\DD\mathfrak{m}_0$, we could have restricted ourselves to
the ordinary algebra, but we decided to opt for the internal one 
to allow for a simple relation to the operation $\iOOO S_{P}$. 
We remark that $\iOOO S_{P0}$ plays a subordinate role as compared to $\iOOO S_{P}$ 
and serves mainly the purpose of elucidating how the ordinary theory is extended by the higher one
in our operational setup. 


With our principal $2$--bundle operational setup in place, 
the ground is set for the operational theory of $2$--connections and $1$-- 
and $2$--gauge transformations presented in II.

\vfil\eject

\section{\textcolor{blue}{\sffamily Operations and principal bundles}}\label{sec:opbrev}

The operational framework is the paradigm on which the whole architecture of the present endeavour
rests. It is indeed a most adequate approach for the study of the total space differential geometry 
of ordinary principal bundles and for this reason, once suitably adapted, holds the potential 
for providing new insights into that of principal $2$--bundles. In this section, we review it
and its application to ordinary principal bundle theory to set notation and terminology
and to make the more technically involved applications of it to principal $2$--bundle theory
more easily accessible. 

In what follows, all algebras considered will be tacitly assumed to be graded commutative, 
unital, associative and real unless otherwise stated.


\subsection{\textcolor{blue}{\sffamily Operations and Lie group spaces: our geometric paradigm}}\label{subsec:opers}

In this subsection, we review the main definitions and results of operation and Lie group space theory.
For a comprehensive treatment, we refer the reader to \ccite{Greub:1973ccc}.
We begin by introducing operations. 


\begin{defi}
An operation $O$ consists of a graded commutative algebra $A$ and a  Lie algebra $\mathfrak{g}$
together with a derivation $d$ of $A$ of degree $1$ and for each element $X\in\mathfrak{g}$ two derivations 
$j_X$, $l_X$ of $A$ of degree $-1$, $0$, respectively, satisfying the six Cartan relations \hphantom{xxxxxxxxxxxxxxxxxxx}
\begin{align}
&[d,d]=0,
\vphantom{\Big]}
\label{opers1}
\\
&[d,j_X]=l_X,
\vphantom{\Big]}
\label{opers2}
\\
&[d,l_X]=0,
\vphantom{\Big]}
\label{opers3}
\\
&[j_X,j_Y]=0,
\vphantom{\Big]}
\label{opers4}
\\
&[l_X,j_Y]=j_{[X,Y]},
\vphantom{\Big]}
\label{opers5}
\\
&[l_X,l_Y]=l_{[X,Y]}
\vphantom{\Big]}
\label{opers6}
\end{align}
for $X,Y\in\mathfrak{g}$, where all commutators are graded. 
\end{defi}

\noindent
We shall denote the operation $O$ by the list $(A,\mathfrak{g},d,j,l)$ 
of its constituent data, or simply $(A,\mathfrak{g})$ when no confusion is possible, 
and say that $O$ is an operation of $\mathfrak{g}$ on $A$. 

\begin{defi}
A morphism $\chi:O\rightarrow O'$ of operations consists of a graded commutative algebra morphism
$\phi:A\rightarrow A'$ and a  Lie algebra morphism $h:\mathfrak{g}'\rightarrow\mathfrak{g}$ 
such that the relations \hphantom{xxxxxxxxxxxxxx}
\begin{align}
&d\hspace{.3pt}{}'\phi=\phi\,d,
\vphantom{\Big]}
\label{opers9}
\\
&j\hspace{.2pt}{}'{}_X\phi=\phi\,j_{h(X)},
\vphantom{\Big]}
\label{opers10}
\\
&l\hspace{.35pt}{}'{}_X\phi=\phi\,l_{h(X)}
\vphantom{\Big]}
\label{opers11}
\end{align}
are obeyed for all $X\in\mathfrak{g}'$. 
\end{defi}

\noindent 
We shall denote 
an operation morphisms $\chi:O\rightarrow O'$ by the list 
$(\phi:A\rightarrow A'$, $h:\mathfrak{g}'\rightarrow\mathfrak{g})$ of its constituent data 
and shall omit the specification of sources and targets if no confusion can arise. 

Operations and operation morphisms thereof constitute a category 
$\bfs{\mathrm{Op}}$.  

Assume that an operation morphism $\chi:O\rightarrow O'$ 
given by the pair $(\phi:A\rightarrow A',h:\mathfrak{g}'\rightarrow\mathfrak{g})$ 
is such that $A=A'$ and $\phi=\id_A$. 
We then have 
\begin{align}
&d\hspace{.3pt}{}'=d,
\vphantom{\Big]}
\label{opers12}
\\
&j\hspace{.2pt}{}'{}_X=j_{h(X)},
\vphantom{\Big]}
\label{opers13}
\\
&l\hspace{.35pt}{}'{}_X=l_{h(X)}
\vphantom{\Big]}
\label{opers14}
\end{align}
for $X\in\mathfrak{g}'$ by \ceqref{opers9}--\ceqref{opers11}. In this case, 
the morphism $\chi$ is fully specified by the underlying Lie algebra morphism $h$. We   
then denote $O'$ by $h^*O$ and call it the pull-back of $O$ by $h$.

Let $O=(A,\mathfrak{g},d,j,l)$ be an operation. Since by \ceqref{opers1} $d^2=0$,
$(A,d)$ is a cochain complex. 

\begin{defi}
The cohomology of $O$ is cohomology of $(A,d)$. 
\end{defi}

\noindent 
$(A,d)$ contains \pagebreak a distinguished subcomplex $(A_{\mathrm{basic}},d)$, 
where $A_{\mathrm{basic}}$, the basic sub\-algebra 
of $A$, consists of all horizontal and invariant elements $a\in A$, 
\begin{align}
&j_Xa=0,
\vphantom{\Big]}
\label{opers7}
\\
&l_Xa=0
\vphantom{\Big]}
\label{opers8}
\end{align}
for $X\in\mathfrak{g}$. 

\begin{defi}
The basic cohomology of $O$ is the cohomology of $(A_{\mathrm{basic}},d)$. 
\end{defi}

Next, we introduce Lie group spaces. 

\begin{defi} \label{defi:lgrsp1}
A Lie group space $S$ is a graded manifold $P$ carrying a right action 
$R:P\times \mathsans{G}\rightarrow P$ of a  Lie group $\mathsans{G}$. 
\end{defi}

\noindent
We shall denote the space $S$ through the list $(P,\mathsans{G},R)$ 
of its defining data. 

\begin{defi} \label{defi:lgrsp2}
A morphism $T:S'\rightarrow S$ of Lie group spaces consists of a map $F:P'\rightarrow P$ of graded 
manifolds and a  Lie group morphism $\eta:\mathsans{G}'\rightarrow\mathsans{G}$ such that 
for $a\in\mathsans{G}'$ one has \hphantom{xxxxxxxxxxxxxxx}
\begin{equation}
R_{\eta(a)}\circ F=F\circ R'{}_a \vphantom{\Big]^f_g}. 
\label{opers15}
\end{equation}
\end{defi}

\noindent
We shall denote a Lie group space morphism $T:S'\rightarrow S$ employing the list 
$(F:P'\rightarrow P,~\eta:\mathsans{G}'\rightarrow\mathsans{G})$ of its constituent data
and shall omit the specification of sources and targets if no confusion is possible. 

Lie group spaces and their morphisms form a category $\bfs{\mathrm{Lsp}}$. 

Consider a Lie group space $S$. For each $X\in \mathfrak{g}$, the Lie algebra of $\mathsans{G}$,  
the vertical vector field $V_X$ of the action $R$ associated with $X$ is defined. \raggedbottom 
The internal function algebra $\iFun(T[1]P)$ (cf. app. \cref{subsec:convnot})
of the shifted tangent bundle $T[1]P$ of $P$ 
is acted upon by the de Rham derivation $d_P$ of degree $1$ 
and for $X\in\mathfrak{g}$ 
by the contraction and Lie derivations $j_{PX}$, $l_{PX}$ 
of $V_X$ of degree $-1$, $0$, all realized as graded vector fields on $T[1]P$.

\begin{prop} \label{prop:lgrsp1}
With any Lie group space $S=(P,\mathsans{G}, R)$ there is associated the 
operation $\iOOO S=(\iFun(T[1]P),\mathfrak{g},d_P,j_P,l_P)$.
\end{prop}

\noindent
The reason why the internal \pagebreak function algebra $\iFun(T[1]P)$ is used instead than the ordinary 
algebra $\Fun(T[1]P)$ is that for the Lie groups spaces considered in this paper
$\Fun(T[1]P)$ is not closed under all operation derivations $j_{PX}$, $l_{PX}$ 
while $\iFun(T[1]P)$ is. For the standard Lie group spaces encountered in geometry, 
$\Fun(T[1]P)$ is closed and thus it is possible and indeed customary to use 
it to construct the operation $\OOO S=(\Fun(T[1]P),\mathfrak{g},d_P,j_P,l_P)$. 
In what follows, we shall consider mostly the operation $\iOOO S$. 


Similarly, morphisms of Lie group spaces induce morphisms of the associated operations 
by the equivariance condition \ceqref{opers15}.

\begin{prop}  \label{prop:lgrsp2}
Let $T:S'\rightarrow S$ be a morphism of Lie group spaces specified by the pair 
$(F:P'\rightarrow P,\eta:\mathsans{G}'\rightarrow\mathsans{G})$. 
Then, $T$ induces a morphism $\iOOO T:\iOOO S\rightarrow\iOOO S'$ of operations 
specified by the pair $(F^*:\iFun(T[1]P)\rightarrow\iFun(T[1]P')$, $h:\mathfrak{g}'\rightarrow\mathfrak{g})$,
where $F^*$ is the  pull-back of function spaces by $F$  and $h$ is the 
Lie algebra morphism associated with $\eta$ by Lie differentiation.
\end{prop}

The map $\iOOO$ that associates to each 
space its operation and to each morphism of spaces its morphism of operations is a functor from the space category 
$\bfs{\mathrm{Lsp}}^{\mathrm{op}}$ into the operation category $\bfs{\mathrm{Op}}$.

Consider the special case of a morphism $T:S'\rightarrow S$ of Lie group spaces specified  by a pair
$(F:P'\rightarrow P,\eta:\mathsans{G}'\rightarrow\mathsans{G})$, where $P=P'$ as graded 
manifolds and $F=\id_P$ as a graded manifold map. Then, by \ceqref{opers15}
\begin{equation}
R'{}_a =R_{\eta(a)}
\label{opers16}
\end{equation}
with $a\in\mathsans{G}'$. We then denote $S'$ as $\eta^*S$ and call it the pull-back
of $S$ by $\eta$. Furthermore, as
the graded algebra morphism $F^*=\id_{\Fun(T[1]P)}$, the operation $\iOOO S'$ is just 
the pull--back $h^*\iOOO S$ of the operation $\iOOO S$ 
by the Lie algebra morphism $h:\mathfrak{g}'\rightarrow\mathfrak{g}$ 
associated with the Lie group morphism $\eta$ by Lie differentiation. Note that 
by construction \hphantom{xxxxxxxxxxxxxxxxxx}
\begin{equation}
\iOOO \eta^*S=h^*\iOOO S.
\label{opers17}
\end{equation}

For a Lie group space $S=(P,\mathsans{G},R)$,  
the cochain complex $(\Fun(T[1]P),d_P)$
is as well--known isomorphic to the de Rham complex $(\Omega^\bullet(P),d_{dR\,P})$. 
When the action $R$ is free, and so the quotient $P/\mathsans{G}$ is a manifold, 
the basic complex $(\Fun(T[1]P)_{\mathrm{basic}},d_P)$ is isomorphic to the de Rham complex 
$(\Omega^\bullet(P/\mathsans{G}),d_{dR\,P/\mathsans{G}})$.
Note however that $(\Fun(T[1]P)_{\mathrm{basic}},d_P)$ 
is defined even when $R$ is not free and this identification strictly speaking fails to hold 
making it possible by way of generalization to meaningfully interpret it in a weaker sense
as the complex $(\Omega^\bullet(P/\mathsans{G}),d_{dR\,P/\mathsans{G}})$.


\subsection{\textcolor{blue}{\sffamily Total space operation of a principal bundle 
}}\label{subsec:ordtotal}

The total space theory of principal bundles, surveyed in this subsection, is concerned 
with the more geometrically intuitive features of principal bundles, in particular its fibered structure
and structure group action. Our expositioin, which follows mainly \ccite{Greub:1973ccc}, 
is organized in such a way to render the naturalness of the operational formulation apparent. 

Let $\mathsans{G}$ be a Lie group. 

\begin{defi}
A principal $\mathsans{G}$--bundle $P$ 
consists of a manifold $P$, a further manifold $M$, a surjective submersion 
$\pi:P\rightarrow M$ and a right $\mathsans{G}$--action $R$ on $P$ 
acting freely and transitively on the fibers of $\pi$. 
\end{defi}

\noindent
$P$, $\mathsans{G}$, $M$ and $\pi$ are called the total space, the structure group, 
the base and the projection of the bundle, respectively. 

In a principal $\mathsans{G}$--bundle $P$, each fiber of $P$ is diffeomorphic 
to $\mathsans{G}$, the orbits of the action $R$ are the fibers themselves 
and the orbit space $P/\mathsans{G}$ is diffeomorphic to the base manifold $M$.
Further, the vertical subbundle $V_P=\ker T\pi$ of the tangent bundle $TP$ of $P$, 
where $T\pi:TP\rightarrow TM$ is the tangent map of the bundle's projection $\pi$, is trivial.  

A connection of 
$P$ is a $\mathsans{G}$--invariant 
distribution $H\subset TP$ that is pointwise transverse to $V_P$. The connection is said to be
flat, if the distribution $H$ is integrable. Connections of $P$ form an affine space $\mathcal{A}_P$.

A gauge transformation of 
$P$ is a $\mathsans{G}$--equivariant 
fiber preserving diffeomorphism $\varPhi$ of $P$. Gauge transformations form a group $\Aut^G_\pi(P)$ under composition,
called the gauge group. $\Aut^G_\pi(P)$ turns out to be isomorphic to the group $\Fun^G(P,\mathsans{G}_{\mathrm{Ad}})$
of $\mathsans{G}$--equivariant maps of $P$ into $\mathsans{G}_{\mathrm{Ad}}$, where $\mathsans{G}_{\mathrm{Ad}}$ is 
$\mathsans{G}$ with the conjugation right action. 

If $H$ is a connection and $\varPhi$ a gauge transformation of $P$, then $\varPhi^{-1}{}_*H$ is also a connection of 
$P$ since the equivariance of $\varPhi$ renders $\varPhi^{-1}{}_*H$  a $\mathsans{G}$--invariant distribution in $TP$. 
In this way, a left action $(\varPhi,H)\rightarrow {}^\varPhi H=\varPhi^{-1}{}_*H$ of the gauge group 
$\Aut^G_\pi(P)$ on the connection space $\mathcal{A}_P$ is defined.

With a principal $\mathsans{G}$--bundle $P$ there is associated a Lie group space $S_P=(P,\mathsans{G},R)$ 
and through this an operation $\OOO S_P=(\Fun(T[1]P),\mathfrak{g},d_P,j_P,l_P)$. 
$\OOO S_P$ furnishes a very elegant and natural differential geometric framework
for the study of connections and gauge transformations of $P$, since 
these are essential defined by the way they behave under the $\mathsans{G}$--action of $P$. 
We shall review the resulting theory in greater detail in sect. 2 
of II. 
Here, in I, we shall use the operational theory of principal bundles just outlined as a prototype
for the corresponding theory of strict principal 2--bundles introduced and 
studied in next section. 


\vfil\eject

\section{\textcolor{blue}{\sffamily Operational total space theory of strict principal 2--bundles}}\label{sec:weiloper}

In this section, using as a model the operational total space theory of ordinary principal bundles reviewed in sect. 
\cref{sec:opbrev}, we present an operation based formulation of the total space theory of strict principal 2--bundles. 
The construction expounded below relies in an essential way on the description of the bundle's structure 
$2$--group as a Lie group crossed module and involves a synthetic recasting of the categorical 
formulation of the theory that leads one into the realm of graded differential geometry. 
The end result is a compact, elegant formulation closely related to the superfield formalism 
employed in many areas of quantum field theory. 



\subsection{\textcolor{blue}{\sffamily Total space theory of principal 2--bundles}}\label{subsec:2prinbund}

Since the goal we set ourselves is generalizing the operation based theory of ordinary principal bundles 
to strict principal $2$--bundles, a suitable total space theory of such $2$--bundles is required. 
This has been worked out 
by Bartels \ccite{Bartels:2006hgtb} and Wockel \ccite{Wockel:2008tspb} (see also \ccite{Nikolaus:2013fvng})
systematically categorifying the standard notion of principal bundle. 
In this subsection, we preset a review of this topic  
based mainly on \ccite{Wockel:2008tspb}. Here, we claim neither originality nor completeness  
and mathematical rigour and shall restrict ourselves to provide those basic notions which 
are required to justify the constructions presented in 
this paper.

\begin{defi}
A strict Lie $2$--group $\hat{\matheul{K}}$ is a group object in the category $\bfs{\mathrm{DiffCat}}$ of smooth categories.
$\hat{\matheul{K}}$ consists thus of the following data.
\begin{enumerate}

\item  A Lie groupoid $\hat{\mathsans{K}}$. 

\item A smooth multiplication functor 
$\hat{\varkappa}:\hat{\mathsans{K}}\times\hat{\mathsans{K}}\rightarrow\hat{\mathsans{K}}$.

\item A smooth inversion 
functor $\hat{\iota}:\hat{\mathsans{K}}\rightarrow\hat{\mathsans{K}}$.

\item A distinguished object $\hat 1$.

\end{enumerate}
$\hat{\varkappa}$, $\hat{\iota}$ and $\hat 1$ obey further the usual axioms of group theory 
at the functor level. 
\end{defi}

\noindent
Under mild assumptions, a strict Lie $2$--group $\hat{\matheul{K}}$ 
can be described equivalently by the following set of data.
\begin{enumerate}

\item A Lie group $\hat{\mathsans{K}}$. 

\item A Lie subgroup $\hat{\mathsans{K}}_0$ of $\hat{\mathsans{K}}$.


\item A Lie groupoid structure $\xymatrix@C-=0.5cm{\hat{\mathsans{K}}\ar@<-2pt>[r]\ar@<+2pt>[r]&\hat{\mathsans{K}}_0}$
all of whose structure maps are Lie group morphisms with the identity assigning map 
$\xymatrix@C-=0.5cm{\hat{\mathsans{K}}_0\ar[r]&\hat{\mathsans{K}}}$ 
being the inclusion map $\hat{\mathsans{K}}_0\subset \hat{\mathsans{K}}$. 

\end{enumerate}
Here, we describe the groupoid $\hat{\mathsans{K}}$ underlying $\hat{\matheul{K}}$ 
through its morphism and object manifolds $\hat{\mathsans{K}}$ and $\hat{\mathsans{K}}_0$. 
With a convenient and harmless abuse of notation, we denote in the same way both the groupoid and 
its morphism manifold, since this latter obviously by itself supports the whole $2$--group structure. 
For this reason, we shall often identify $\hat{\matheul{K}}$ through the group $\hat{\mathsans{K}}$
or more explicitly the group pair $\hat{\mathsans{K}}$, $\hat{\mathsans{K}}_0$. 
Whenever necessary, we shall mention the relevant structure maps.





\begin{defi}
A strict principal $\hat{\matheul{K}}$--$2$--bundle $\mathcal{P}$ is a principal bundle object 
in the category $\bfs{\mathrm{DiffCat}}$ 
of smooth categories. 
$\mathcal{P}$ consists thus of the following data.
\begin{enumerate}

\item A Lie groupoid $\hat{{P}}$. 

\item A discrete smooth category ${M}$. 

\item A surjective submersion smooth projection functor $\hat{{\pi}}:\hat{{P}}\rightarrow{M}$. 

\item A strict Lie $2$--group $\hat{\mathsans{K}}$. 

\item A smooth right $\hat{\mathsans{K}}$--action functor 
$\hat{{R}}:\hat{{P}}\times\hat{\mathsans{K}}\rightarrow\hat{{P}}$
such that 
\begin{equation}
\hat{{\pi}}\circ\hat{{R}}=\hat{{\pi}}\circ\mathrm{pr}_1
\label{2prinbund*}
\end{equation}
strictly on the nose. 

\item For each small enough open neighborhood ${U}\subset{M}$, two
reciprocally weak\-ly inverse $\hat{\mathsans{K}}$--equivariant trivializing smooth functors 
$\hat{\varPhi}_{{U}}:\hat{{\pi}}^{-1}({U})\rightarrow{U}\times\hat{\mathsans{K}}$ \linebreak and 
$\tilde{\hat{\varPhi}}_{{U}}:{U}\times\hat{\mathsans{K}}\rightarrow\hat{{\pi}}^{-1}({U})$ such that \pagebreak 
\begin{align}
&\mathrm{pr}_{\,{U}}\circ\hat{\varPhi}_{{U}}=\hat{{\pi}}\big|_{\hat{{\pi}}^{-1}({U})},
\vphantom{\Big]}
\label{2prinbund5}
\\
&\hat{{\pi}}\circ\tilde{\hat{\varPhi}}_{{U}}=\mathrm{pr}_{\,{U}}  \vspace{-1cm}
\vphantom{\Big]}
\label{2prinbund6}
\end{align}
strictly on the nose, where ${U}\times\hat{\mathsans{K}}$ is endowed with the 
right $\hat{\mathsans{K}}$--action functor given by right $\hat{\mathsans{K}}$--multiplication on the 
factor $\hat{\mathsans{K}}$. 

\end{enumerate} 
\end{defi}

\noindent
The adjective 'semistrict' is used instead of 'strict' in \ccite{Wockel:2008tspb}. 


We can unpack the above definition and analyze a strict principal $\hat{\matheul{K}}$--$2$--bundle 
$\hat{\mathcal{P}}$ in more explicit terms under mild assumptions. Though this is a straightforward task, we decided to 
review it in some detail to make the reading of the rest of the paper easier. 

The groupoid structure of $\hat{\mathcal{P}}$ is described by the following elements.
\begin{enumerate} 

\item A smooth manifold $\hat{P}$. 

\item A submanifold $\hat{P}_0$ of $\hat{P}$.

\item A Lie groupoid structure $\xymatrix@C-=0.5cm{\hat{P}\ar@<-2pt>[r]\ar@<+2pt>[r]&\hat{P}_0}$
whose identity assigning map $\xymatrix@C-=0.5cm{\hat{P}_0\ar[r]&\hat{P}}$ 
is the inclusion map $\hat{P}_0\subset \hat{P}$. 

\end{enumerate}

\noindent 
Here, analogously to the way we did for $2$--groups, we characterize the groupoid $\hat P$ underlying $\hat{\mathcal{P}}$ 
through its morphism and object manifolds $\hat P$ and $\hat P_0$, identifying further the groupoid 
the former. The base structure of $\hat{\mathcal{P}}$ is given by an ordinary manifold, adding a further 
element.
\begin{enumerate}[resume]

\vspace{.33mm}

\item A manifold $M$. 

\vspace{.33mm}

\end{enumerate}

\noindent 
We reduce so the discrete smooth category $M$ underlying $\hat{\mathcal{P}}$, which has only identity morphisms, 
to its object manifold, which we also denote as $M$. The projection structure of $\hat{\mathcal{P}}$ is consequently given 
by an ordinary map. 
\begin{enumerate}[resume]

\vspace{.33mm}

\item A smooth surjective submersion projection map $\hat{\pi}:\hat{P}\rightarrow M$. 

\vspace{.33mm}

\end{enumerate}

\noindent 
In line with the conventions we are outlining, here we describe the projection functor $\hat\pi$ of
$\hat{\mathcal{P}}$ through its action on the morphism manifold $\hat P$ 
also denoted as $\hat\pi$. $\hat{\pi}$ restricts to a smooth surjective submersion projection map 
$\hat\pi_0:\hat P_0\rightarrow M$ giving the functor's action on the object manifold $\hat P_0$. 
By functoriality, $\hat\pi$ and $\hat\pi_0$ satisfy certain relations. We mention only that 
\begin{equation}
\hat{\pi}=\hat{\pi}_0\circ\hat s=\hat{\pi}_0\circ\hat t, 
\label{2prinbund10}
\end{equation}
where $\hat s$, $\hat t$ are the source and target map of the groupoid $\hat P$. 
The structure $2$--group of $\hat{\mathcal{P}}$ is analyzed next as follows. 
\begin{enumerate}[resume]

\item A Lie group $\hat{\mathsans{K}}$. 

\item A Lie subgroup $\hat{\mathsans{K}}_0$ of $\hat{\mathsans{K}}$.

\item A Lie $2$--group structure $\xymatrix@C-=0.5cm{\hat{\mathsans{K}}\ar@<-2pt>[r]\ar@<+2pt>[r]&\hat{\mathsans{K}}_0}$

whose identity assigning map $\xymatrix@C-=0.5cm{\hat{\mathsans{K}}_0\ar[r]&\hat{\mathsans{K}}}$ 
is the inclusion map $\hat{\mathsans{K}}_0\subset \hat{\mathsans{K}}$. 

\end{enumerate}

\noindent 
We explained the reason for this earlier. The right $\hat{\mathsans{K}}$--action structure
of $\hat{\mathcal{P}}$ is given in terms of a further element.
\begin{enumerate}[resume]

\item A smooth right $\hat{\mathsans{K}}$--action map $\hat R:\hat{{P}}\times\hat{\mathsans{K}}\rightarrow\hat{{P}}$
such that 
\begin{equation}
\hat{{\pi}}\circ\hat{{R}}=\hat{{\pi}}\circ\mathrm{pr}_1
\label{2prinbund*/1}
\end{equation}

\end{enumerate}

\noindent 
(cf. eq. \ceqref{2prinbund*}). Again, we describe the $\hat{\mathsans{K}}$--action functor $\hat R$ of
$\hat{\mathcal{P}}$ through its action on the morphism manifold $\hat P$, 
which we also denote as $\hat R$. $\hat R$ restricts to a smooth right $\hat{\mathsans{K}}_0$--action map 
$\hat R_0:\hat P_0\times\hat{\mathsans{K}}_0\rightarrow\hat P_0$ 
giving the functor's action on the object manifold $\hat P_0$. 
$\hat R$ and $\hat R_0$ obey certain relations by virtue of functoriality. 
Finally, by the property of local trivializability, a full collection of local 
trivializing maps is available. 
\begin{enumerate}[resume]

\item For any small open neighborhood $U$ of $M$, two trivializing $\hat{\mathsans{K}}$--equivariant smooth maps 
$\hat{\varPhi}_U:\hat{\pi}^{-1}(U)\rightarrow U\times\hat{\mathsans{K}}$ 
and $\tilde{\hat{\varPhi}}_U:U\times\hat{\mathsans{K}}\rightarrow\hat{\pi}^{-1}(U)$ such that 
\begin{align}
&\mathrm{pr}_{\,{U}}\circ\hat{\varPhi}_{{U}}=\hat{{\pi}}\big|_{\hat{{\pi}}^{-1}({U})},
\vphantom{\Big]}
\label{2prinbund12}
\\
&\hat{{\pi}}\circ\tilde{\hat{\varPhi}}_{{U}}=\mathrm{pr}_{\,{U}}  \vspace{-1cm}
\vphantom{\Big]}
\label{2prinbund13}
\end{align}

\end{enumerate}

\noindent
(cf. eqs. \ceqref{2prinbund5}, \ceqref{2prinbund6}). 
Analogously to before, we describe the trivializing functors $\hat{\varPhi}_U$ and $\tilde{\hat{\varPhi}}_U$ 
of $\hat{\mathcal{P}}$ on $U$ \pagebreak through their actions on the morphism manifold $\hat P$ 
denoted as $\hat{\varPhi}_U$ and $\tilde{\hat{\varPhi}}_U$ too. 
$\hat{\varPhi}_U$ and $\tilde{\hat{\varPhi}}_U$ restrict to 
trivializing $\hat{\mathsans{K}}_0$--equivariant smooth maps 
$\hat{\varPhi}_{U0}:\hat{\pi}_0{}^{-1}(U)\rightarrow U\times\hat{\mathsans{K}}_0$ 
and $\tilde{\hat{\varPhi}}_{U0}:U\times\hat{\mathsans{K}}_0\rightarrow\hat{\pi}_0{}^{-1}(U)$
giving the functors' actions on the object manifolds $\hat{\pi}_0{}^{-1}(U)$ and $U\times\hat{\mathsans{K}}_0$. 
$\hat{\varPhi}_U$ and $\tilde{\hat{\varPhi}}_U$  and $\hat{\varPhi}_{U0}$ and $\tilde{\hat{\varPhi}}_{U0}$
obey a large set of relations stemming from functoriality and reciprocal weak invertibility. 
Further, there are a 
$\hat{\mathsans{K}}$--equivariant map $\hat T_U:\hat{\pi}^{-1}(U)\rightarrow\hat{\mathsans{K}}$, where
$\hat{\mathsans{K}}$ is endowed with the right multiplication action, such that 
\begin{equation}
\hat{\varPhi}_U(X)=(\hat{\pi}(X),\hat T_U(X))
\label{2prinbund14}
\end{equation}
for $X\in\hat{\pi}^{-1}(U)$ and for each $m\in U$ 
a $\hat{\mathsans{K}}$--equivariant map $\hat{X}_{Um}:\hat{\mathsans{K}}\rightarrow\hat{\pi}^{-1}(U)$
with the property that \hphantom{xxxxxxxxxxxxxxxxx}
\begin{equation}
\tilde{\hat{\varPhi}}_U(m,A)=\hat{X}_{Um}(A)
\label{2prinbund16}
\end{equation}
for $A\in\hat{\mathsans{K}}$ obeying many further relations deriving from $\hat{\varPhi}_U$ and $\tilde{\hat{\varPhi}}_U$
being weakly inverse functors. The associated $\hat{\mathsans{K}}_0$--equivariant maps 
$\hat{T}_{U0}:\hat{\pi}_0{}^{-1}(U)\rightarrow\hat{\mathsans{K}}_0$ and 
$\hat{X}_{U0m}:\hat{\mathsans{K}}_0\rightarrow\hat{\pi}_0{}^{-1}(U)$ are again restrictions.


The fact that the local trivializing functors $\hat{\varPhi}_{U}$ and $\tilde{\hat{\varPhi}}_{U}$ 
of a strict principal $2$--bundle $\hat{\mathcal{P}}$ are only reciprocally weakly inverse 
makes $\hat{\mathcal{P}}$ unlike an ordinary principal bundle. 
In general, $\hat{\varPhi}_U$ is not injective and one can use it to coordinatize
$\hat{\pi}^{-1}(U)$ only in a weaker sense than in the ordinary theory. 
Likewise, $\tilde{\hat{\varPhi}}_U$ is generally non surjective and one can use it to parametrize 
$\hat{\pi}^{-1}(U)$ only up to suitable isomorphism.

Because of the above properties, 
the $\hat{\mathsans{K}}$--action $\hat{R}$ on $\hat P$ is fiberwise free but fiberwise 
transitive only up to appropriate isomorphism. Furthermore, 
a $\hat{\mathsans{K}}$--invariant map $f:\hat{P}\rightarrow\mathbb{R}$ cannot be identified with 
a map $\bar f:M\rightarrow\mathbb{R}$ in general. However, as locally on a neighborhood $U$ $\hat{\pi}^{-1}(U)$ is
$\hat{\mathsans{K}}$--equivariantly equivalent to $U\times\hat{\mathsans{K}}$, which 
is an ordinary principal $\hat{\mathsans{K}}$--bundle, a $\hat{\mathsans{K}}$--invariant map 
$f:\hat{\pi}^{-1}(U)\rightarrow\mathbb{R}$ can be described as a map 
$\bar f_U:U\rightarrow\mathbb{R}$. However, for fixed $U$,
$\bar f_U$ is built through and generally depends on the underlying trivializing functor
$\tilde{\hat{\varPhi}}_{U}$ and is not therefore uniquely associated to $f$. 
With the appropriate notion of horizontality, similar statements
hold also for differential forms. 

To the best of our knowledge, \pagebreak there does not exist at the moment any definition of connection on a 
strict principal $2$--bundle akin to that of ordinary principal bundle theory 
formulated in terms of a horizontal invariant distribution in the tangent bundle of the bundle. 
There exits however a definition of gauge transformation analogous to that of the  ordinary theory 
as an equivariant fiber preserving bundle automorphism.

A weak $2$--group $\matheul{W}$ is a group object in the category $\bfs{\mathrm{Cat}}$ of categories, that is 
a groupoid $\mathsans{W}$ endowed with a multiplication functor 
$\varkappa:\mathsans{W}\times\mathsans{W}\rightarrow\mathsans{W}$, an inversion 
functor $\iota:\mathsans{W}\rightarrow\mathsans{W}$ and a distinguished object $1$
obeying the usual axioms of group theory 
at the functor level up to natural isomorphisms satisfying suitable coherence conditions. 
Strict $2$--groups are those weak $2$--groups for which all these natural isomorphisms are identities. 
Replacing the category  $\bfs{\mathrm{Cat}}$ with $\bfs{\mathrm{DiffCat}}$, one defines similarly 
the notion of weak Lie $2$--group and strict Lie $2$--group, the latter of which was introduced earlier. 

A given category ${C}$ is characterized by the weak $2$--group $\Aut({C})$ whose objects are the 
weakly invertible functors ${F}:{C}\rightarrow{C}$, whose morphisms are the natural 
isomorphisms ${\beta}:{F}\Rightarrow {G}$ and whose $2$--group structure is as follows. 
The group structure of $\Aut({C})_0$ is given by functor composition and weak inversion,
the group structure of $\Aut({C})_1$ is given by the so called Godement composition and inversion
of natural isomorphisms, the groupoid structure of $\Aut({C})_1$ consists in the usual composition 
and inversion of natural isomorphisms. 

There is a notion of gauge $2$--group of a principal $2$--bundle
extending the familiar notion of gauge group of the ordinary theory. 
The objects and morphisms of the gauge $2$--group 
are closely related to the $1$-- and $2$--gauge transformations in 
higher gauge theory and therefore are of considerable interest for us.
There are a few categorical equivalent ways of presenting the gauge $2$--group,
which we review next. Below, we let $\hat{\mathcal{P}}$ be a strict principal 
$\hat{\matheul{K}}$--$2$--bundle with base $M$. 

The gauge $2$--group of $\hat{\mathcal{P}}$ can be characterized as its automorphism $2$--group.  

\begin{defi}
The automorphism $2$--group of $\hat{\mathcal{P}}$, $\Aut_{\hat\pi}^{\hat{\mathsans{K}}}\!(\hat P)$,
consists in the weak sub--$2$--group of $\Aut(\hat{{P}})$ of $\hat{\mathsans{K}}$--equivariant 
weakly invertible functors ${F}:\hat{{P}}\rightarrow\hat{{P}}$ preserving the bundle 
projection $\hat{\pi}$, \hphantom{xxxxxxxxxxxxxxxxxxxx}
\begin{equation}
\hat{{\pi}}\circ {F}=\hat{{\pi}}
\label{2prinbund17}
\end{equation}
strictly, and $\hat{\mathsans{K}}$--equivariant natural isomorphisms ${\beta}:{F}\Rightarrow {G}$. 
\end{defi}

\noindent
The gauge $2$--group of $\hat{\mathcal{P}}$ can be characterized also as its
equivariant structure $2$--group valued morphism $2$--group. 

\begin{defi}
The equivariant structure $2$--group valued morphism $2$--group 
of $\hat{\mathcal{P}}$, $\Fun^{\hat{\mathsans{K}}}\!(\hat{{P}},\hat{\mathsans{K}}_{\mathrm{Ad}})$, is 
the strict $2$--group whose objects and morphisms are respectively
the $\hat{\mathsans{K}}$--equivariant functors ${F}:\hat{{P}}\rightarrow\hat{\mathsans{K}}$ 
and $\hat{\mathsans{K}}$--equivariant natural isomorphisms ${\beta}:{F}\Rightarrow {G}$ and 
whose $2$--group operations are 
those induced pointwise by those of $\hat{\mathsans{K}}$, 
where $\hat{\mathsans{K}}_{\mathrm{Ad}}$ denotes $\hat{\mathsans{K}}$
with the conjugation right action. 
\end{defi}

\noindent 
The following theorem, shown in ref. \ccite{Wockel:2008tspb}, extends a basic property 
of the gauge groups of ordinary principal bundles to principal $2$--bundles. 

\begin{theor}
The automorphism $2$--group $\Aut_{\hat{\pi}}^{\hat{\mathsans{K}}}\!(\hat{{P}})$ 
and the equivariant structure $2$--group valued morphism $2$--group
$\Fun^{\hat{\mathsans{K}}}\!(\hat{{P}},\hat{\mathsans{K}}_{\mathrm{Ad}})$
are equivalent as weak 2-groups. 
\end{theor}

\noindent
In this way, in the appropriate categorical sense, we can identify $\Aut_{\hat{\pi}}^{\hat{\mathsans{K}}}\!(\hat{{P}})$ 
and $\Fun^{\hat{\mathsans{K}}}\!(\hat{{P}},\hat{\mathsans{K}}_{\mathrm{Ad}})$. 
This second realization of the gauge $2$--group of $\hat{\mathcal{P}}$ is more convenient
and we shall refer to it in the following. 
Under the assumptions that ${M}$ is com\-pact, $\hat{\mathsans{K}}$ is locally exponentiable
and the $\hat{\mathsans{K}}$--action is principal on $\hat P$ and $\hat P_0$,
one can prove  that 
$\Fun^{\hat{\mathsans{K}}}\!(\hat{{P}},\hat{\mathsans{K}}_{\mathrm{Ad}})$ is a strict Lie $2$--group. 
Moreover, 
$\Fun^{\hat{\mathsans{K}}}\!(\hat{{P}},\hat{\mathsans{K}}_{\mathrm{Ad}})$ enjoys a simple concrete 
description. 
\begin{enumerate}

\item The objects of $\Fun^{\hat{\mathsans{K}}}\!(\hat{{P}},\hat{\mathsans{K}}_{\mathrm{Ad}})$ 
 are $\hat{\mathsans{K}}$--equivariant 
functors $F:\hat{P}\rightarrow\hat{\mathsans{K}}_{\mathrm{Ad}}$. $\vphantom{\ul{\ul{\ul{\ul{g}}}}}$


\item The morphisms of $\Fun^{\hat{\mathsans{K}}}\!(\hat{{P}},\hat{\mathsans{K}}_{\mathrm{Ad}})$ 
from an object $F$ to another $G$ are the
$\hat{\mathsans{K}}_0$--equivariant maps $\beta:\hat{P}_0\rightarrow\hat{\mathsans{K}}_{\mathrm{Ad}}$ 
such that $\beta(x):F_0(x)\rightarrow G_0(x)$ for $x\in\hat{P}_0$. 

\end{enumerate}

\noindent
The $\hat{\mathsans{K}}_0$--equivariant map 
$F_0:\hat{P}_0\rightarrow\hat{\mathsans{K}}_{0\mathrm{Ad}}$ associated with the object $F$ is again 
a restriction. 

A useful local description of the gauge $2$--group of $\hat{\mathcal{P}}$  is also 
available. On a trivializing neighborhood  ${U}\subset{M}$, where $\hat{\pi}^{-1}({U})$ is 
$\hat{\mathsans{K}}$--equivariantly equivalent to ${U}\times\hat{\mathsans{K}}$,
the $2$--groups $\Fun^{\hat{\mathsans{K}}}\!(\hat\pi^{-1}(U),\hat{\mathsans{K}}_{\mathrm{Ad}})$ 
and  $\Fun^{\hat{\mathsans{K}}}\!({U}\times\hat{\mathsans{K}},\hat{\mathsans{K}}_{\mathrm{Ad}})$
are equivalent as weak $2$--groups. The latter $2$--group has in turn the following description. 

\begin{defi}
For a neighborhood ${U}\subset{M}$,
we let $\Fun(\,{U},\hat{\mathsans{K}})^{\mathrm{hop}}$ be the strict $2$-- group 
whose objects are functors ${V}:{U}\rightarrow\hat{\mathsans{K}}$, whose morphisms are natural 
isomorphisms ${\sigma}:V\Rightarrow W$ and whose $2$--group operations are those induced pointwise by 
those of $\hat{\mathsans{K}}^{\mathrm{hop}}$, the $2$ group $\hat{\mathsans{K}}$ with the opposite multiplication functor. 
\end{defi}

\begin{prop}
The $2$--group $\Fun^{\hat{\mathsans{K}}}\!(\,{U}\times\hat{\mathsans{K}},\hat{\mathsans{K}}_{\mathrm{Ad}})$ is 
isomorphic as a strict $2$--group to the $2$--group $\Fun(\,{U},\hat{\mathsans{K}})^{\mathrm{hop}}$ 
\end{prop}

\noindent
The $2$--group $\Fun(\,{U},\hat{\mathsans{K}})^{\mathrm{hop}}$  
has a very simple structure strongly reminiscent that of the gauge group of 
ordinary principal bundles.
\begin{enumerate}

\item An object ${V}$ of $\Fun(\,{U},\hat{\mathsans{K}})$
is specified by a map $V:U\rightarrow\hat{\mathsans{K}}_0$. 

\item A morphism ${\sigma}:V\Rightarrow W$ of $\Fun(\,{U},\hat{\mathsans{K}})$ is similarly specified by a map
$\sigma:M\rightarrow\hat{\mathsans{K}}$ such that $\sigma(m):V(m)\rightarrow W(m)$ for $m\in U$. 

\end{enumerate}
\noindent
Under the strict $2$--group isomorphism $\Fun({U},\hat{\mathsans{K}})^{\mathrm{hop}}\simeq
\Fun^{\hat{\mathsans{K}}}\!(\,{U}\times\hat{\mathsans{K}},\hat{\mathsans{K}}_{\mathrm{Ad}})$, the object ${H}_{{V}}$ of 
$\Fun^{\hat{\mathsans{K}}}\!({U}\times\hat{\mathsans{K}},\hat{\mathsans{K}}_{\mathrm{Ad}})$ corresponding to an object
${V}$ of $\Fun({U},\hat{\mathsans{K}})$ is given by the expression \hphantom{xxxxxxxxxxxxxxxx}
\begin{equation}
H_V(m,A)=A^{-1}V(m)A
\label{2prinbund18}
\end{equation}
with $m\in U$ and $A\in\hat{\mathsans{K}}$. The morphism ${\theta}_{{\sigma}}:{H}_{{V}}\rightarrow {H}_{{W}}$ 
of $\Fun^{\hat{\mathsans{K}}}\!(\,{U}\times\hat{\mathsans{K}},\hat{\mathsans{K}}_{\mathrm{Ad}})$ 
corresponding to a morphism 
${\sigma}:{V}\Rightarrow {W}$ of $\Fun(\,{U},\hat{\mathsans{K}})$ is similarly specified by 
\begin{equation}
\theta_\sigma(m,a)=a^{-1}\sigma(m)a
\label{2prinbund19}
\end{equation}
with $m\in U$ and $a\in\hat{\mathsans{K}}_0$. 

In the next subsection, we shall introduce the synthetic description of principal $2$--bundles,
which brings the categorical formulation expounded above closer to higher gauge theory.


\subsection{\textcolor{blue}{\sffamily Synthetic formulation of  principal 2--bundle theory}}\label{subsec:synth}

The total space theory of principal $2$--bundles expounded in subsect. \cref{subsec:2prinbund} is 
elegant and geometrically intuitive, but it has certain shortcomings. In such framework, no 
viable definition of $2$--connection and $1$--gauge transformation is available that is 
satisfactorily close to the corresponding notions of higher gauge theory, which is our main concern.
(See however refs. \ccite{Waldorf:2016tsct,Waldorf:2017ptpb} for an interesting attempt
in this direction.) As we shall show in II, 
these problems can be solved by switching from the categorical framework
to a distinct but related description, which we shall call synthetic, for lack of a better term 
and with an abuse of language, because the notions it is based 
on are akin in spirit to those of synthetic differential geometry. 
The synthetic approach turns out to be closely related to the standard formulation of strict higher gauge theory,
as we require.

It is difficult to justify the change of perspective we are espousing in simple intuitive terms. 
The synthetic formulation is essentially validated {\it a posteriori} by the viability of the results
it leads to. 

A fundamental property of strict $2$--groups is their equivalence to group crossed modules
described in greater detail in subsect. \cref{subsec:liecm}. 

\begin{theor} \label{theor:2lgr2cm}
For a strict Lie $2$--group $\hat{\matheul{K}}$
with source and target maps $\hat{s}$, $\hat{t}$, there exists a Lie group isomorphism \hphantom{xxxxxxxxxxx}
\begin{equation}
\hat{\mathsans{K}}\simeq\ker\hat{s}\rtimes_{\hat{\lambda}}\hat{\mathsans{K}}_0, 
\label{2prinbund3}
\end{equation}
where 
the semidirect product group structure is defined with respect to the left 
action $\hat{\lambda}:\hat{\mathsans{K}}_0\times\ker\hat{s}\rightarrow\ker\hat{s}$ given by
\begin{equation}
\hat{\lambda}(a,Z)=aZa^{-1} 
\label{2prinbund4}
\end{equation}
with $a\in\hat{\mathsans{K}}_0$, $Z\in\ker\hat{s}$. 
Furthermore, under the isomorphism \ceqref{2prinbund3},
the Lie groupoid structure of $\hat{\mathsans{K}}$ can be expressed completely
in terms of the group structures of $\hat{\mathsans{K}}_0$ and $\ker\hat{s}$, the action map $\hat{\lambda}$ 
and the target map \hphantom{xxxxxxxxxxxxxx} 
\begin{equation}
\hat{\epsilon}=\hat{t}\big|_{\ker\hat{s}}.
\label{2prinbund4/1}
\end{equation}
\end{theor}

Viewing the structure $2$--group $\hat{\matheul{K}}$ of the relevant principal $2$--bundle $\hat{\mathcal{P}}$
as a Lie group crossed module $(\hat{\mathsans{K}}_0,\ker\hat{s})$ 
leads to a geometric framework of the higher theory closer in form and in spirit to that of the ordinary 
one. \pagebreak It is also a necessary starting point for the synthetic formulation, 
which essentially amounts to trading $\ker\hat{s}$ with $\mathrm{Lie\,ker}\hat{s}[1]$.

\begin{defi} \label{defi:synthk}
The synthetic form $\mathsans{K}$ of the morphism group $\hat{\mathsans{K}}$ is the graded Lie group 
of the internal functions $Q\in\iMap(\mathbb{R}[-1],\hat{\mathsans{K}})$ of the form 
\begin{equation}
Q(\bar\alpha)=\ee^{\bar\alpha L}a, \qquad \bar\alpha\in\mathbb{R}[-1]
\label{ex2lgr2cm1}
\end{equation}
with $a\in\hat{\mathsans{K}}_0$, $L\in\mathrm{Lie\,ker}\hat{s}[1]$.
The synthetic form $\mathsans{K}{}_0$ of the object group $\hat{\mathsans{K}}_0$ is the graded Lie subgroup 
of $\mathsans{K}$ constituted by the elements 
of the special form $q(\bar\alpha)=a$. 
\end{defi}

\noindent
The group operations are pointwise multiplication and inversion. 
$\mathsans{K}_0$ can be identified canonically with $\hat{\mathsans{K}}_0$. 

The synthetic groups $\mathsans{K}$, $\mathsans{K}_0$ stem from 
the morphism and object groups $\hat{\mathsans{K}}$,
$\hat{\mathsans{K}}_0$ of the strict Lie $2$--group $\hat{\matheul{K}}$. 
Furthermore, $\mathsans{K}_0$ is a subgroup of $\mathsans{K}$. 
However, $\mathsans{K}$, $\mathsans{K}_0$ are not the morphism and object groups of 
any synthetic strict Lie $2$--group $\matheul{K}$, as not all the groupoid structure maps of $\hat{\mathsans{K}}$
can be directly extended to $\mathsans{K}$. A graded Lie group 
morphism $s:\mathsans{K}\rightarrow\mathsans{K}_0$ extending the source morphism $\hat s$ of $\hat{\mathsans{K}}$
does in fact exist. 

\begin{defi}
The synthetic form of $\hat s$ is the element $s\in\Hom(\mathsans{K},\mathsans{K}_0)$ given by 
\begin{equation}
s(Q)=\hat s\circ Q 
\label{ex2lgr2cm5}
\end{equation}
with $Q\in\mathsans{K}$.
\end{defi}

\noindent 
$s$ is well--defined, since,  writing $Q$ as in \ceqref{ex2lgr2cm1}, $s$ is given by 
$s(Q)(\bar\alpha)=\hat s(Q(\bar\alpha))=a$ and so has a range lying in $\mathsans{K}_0$ as required. 
A partner morphism $t\in\Hom(\mathsans{K},\mathsans{K}_0)$ extending the target morphism 
$\hat t$ of $\hat{\mathsans{K}}$ instead does not. The would--be  target map $t$ would read 
as $t(Q)(\bar\alpha)=\hat t(Q(\bar\alpha))=
\ee^{\bar\alpha \dot{\hat t}(L)}a$, where $\dot{\hat t}$ is the Lie differential of $\hat t$, 
and thus have a range generally lying outside $\mathsans{K}_0$. Consequently, 
composition also cannot be defined in $\mathsans{K}$ and $\mathsans{K}$ has no groupoid structure. 


\begin{defi} \label{defi:synthp}
The synthetic form $P$ of the morphism manifold $\hat P$ is the graded manifold 
of the internal functions $V\in\iMap(\mathbb{R}[-1],\hat P)$ of the form 
\begin{equation}
V(\bar\alpha)=\hat R(x,Q(\bar\alpha)), \qquad \bar\alpha\in\mathbb{R}[-1]
\label{ex2lgr2cm2}
\end{equation}
with $x\in\hat P_0\subset \hat P$, $Q\in\mathsans{K}$. 
The synthetic form $P_0$ of the object manifold $\hat P_0$ is the graded submanifold of $P$ constituted by 
of the points 
of the special form $v(\bar\alpha)=x$. 
\end{defi}

\noindent 
Since $\iMap(\mathbb{R}[-1],\hat P)\simeq T[1]\hat P$, the shifted tangent bundle of $\hat P$,  
$P$ is a vector subbundle of $T[1]\hat P$. Indeed, by the way it is defined and in analogy to the ordinary notion, 
$P$ can be described as $V[1]\hat P|_{\hat P_0}$, the shifted vertical subbundle of $\hat P$ restricted to $\hat P_0$.
$P_0$ can be identified canonically with $\hat P_0$, 
which in turn can be regarded as the zero section of $V[1]\hat P|_{\hat P_0}$. 

The synthetic manifolds $P$, $P_0$ answer to the morphism and object manifolds 
$\hat P$, $\hat P_0$ of the $2$--bundle $\hat{\mathcal{P}}$. $P_0$ is further a submanifold of $P$. 
However, $P$, $P_0$ are not the morphism and object manifolds of any synthetic principal $\matheul{K}$--$2$--bundle
$\mathcal{P}$. To begin with, the synthetic groups $\mathsans{K}$, $\mathsans{K}_0$ cannot constitute the $2$--bundle's 
structure $2$--group $\matheul{K}$ for reasons explained earlier. Moreover, not all the groupoid structure maps of $\hat P$
can be consistently extended to $P$. A map $s:P\rightarrow P_0$ extending the source map $\hat s$ of $\hat P$ 
is in fact available. 

\begin{defi}
The synthetic form of $\hat s$ is the map $s\in\Map(P,P_0)$ given by 
\begin{equation}
s(V)=\hat s\circ V 
\label{ex2lgr2cm6}
\end{equation}
with $V\in P$. 
\end{defi}

\noindent
$s$ is well--defined, as, for $V\in P$ of the form \ceqref{ex2lgr2cm2} with $Q\in \mathsans{K}$ 
as in \ceqref{ex2lgr2cm1}, $s$ reads as $s(V)(\bar\alpha)=\hat s(V(\bar\alpha))
=\hat s(\hat R(x,Q(\bar\alpha)))=\hat R_0(\hat s(x),\hat s(Q(\bar\alpha)))
=\hat R_0(x,a)$ and so  has a range lying in $P_0$ as wished. 
A mate target map $t\in\Map(P,P_0)$ extending the target map $\hat t$ of $\hat P$ instead 
is not. The supposedly target map $t$ would be given by 
$t(V)(\bar\alpha)=\hat t(V(\bar\alpha))=\hat t(\hat R(x,Q(\bar\alpha)))=\hat R_0(\hat t(x),\hat t(Q(\bar\alpha)))$
$=\hat R_0(x,\ee^{\bar\alpha \dot{\hat t}(L)}a)$ 
and thus have a range lying outside $P_0$.  As a consequence, 
composition also cannot be defined in $P$ and $P$ has no groupoid structure. 
In spite of this findings, many of the properties of $\hat P$ and $\hat P_0$ 
as part of the $2$--bundle $\hat{\mathcal{P}}$ do extend in an appropriate form to $P$ and $P_0$. 
We study these in some detail next. 

Projection maps $\pi$ and $\pi_0$ of $P$ induced by $\hat\pi$ and $\hat\pi_0$  respectively exist
and enjoy the expected properties. 

\begin{defi} \label{defi:pisynth}
The synthetic form of $\hat\pi$ is the map $\pi\in\Map(P,M)$ given by 
\begin{align}
\pi(V)=\hat\pi\circ V,  
\vphantom{\Big]}
\label{ex2lgr2cm3}
\end{align}
with $V\in P$. 
The synthetic form  of $\hat\pi_0$ is the map $\pi_0\in\Map(P_0,M)$ resulting from restricting 
$\pi$ to $P_0$. 
\end{defi}

\noindent
$\pi$ extends $\hat\pi$ because, 
writing $V\in P$ in the form \ceqref{ex2lgr2cm2} above, 
$\pi$ is given by $\pi(V)(\bar\alpha)=\hat \pi(V(\bar\alpha))=\hat\pi(\hat R(x,Q(\bar\alpha)))=\hat\pi_0(x)$
and thus has a range lying in $M$. 

\begin{prop}\label{prop:pisynth}
$\pi$ and $\pi_0$ are surjective submersions. Further, $\pi=\pi_0\circ s$.
\end{prop}

\noindent
Compare the above relation with  \ceqref{2prinbund10}. 

\begin{proof}
Letting $\rho:P\rightarrow\hat P_0$ be the projection map of $P\simeq V[1]\hat P|_{\hat P_0}$,  
one has $\pi=\hat\pi\circ\rho$. As  $\hat\pi$, $\rho$ are both surjective submersions, so is $\pi$. 
Likewise, letting $\rho_0:P_0\rightarrow\hat P_0$ be the canonical identification of $P_0$ and $\hat P_0$,
one has $\pi_0=\hat\pi_0\circ\rho_0$. Since $\hat\pi_0$, $\rho_0$ are both surjective submersions, so is $\pi_0$. 
The relation $\pi_0\circ s=\pi$ follows from the identity $\hat\pi_0\,\circ\,\hat s=\hat\pi$, \ceqref{ex2lgr2cm5}
and  \ceqref{ex2lgr2cm6}.  
\end{proof}

The right actions $\hat R$ and $\hat R_0$ of $\hat{\mathsans{K}}$ and $\hat{\mathsans{K}}_0$ on $\hat P$
and $\hat P_0$ similarly give rise to right actions $R$ and $R_0$ of $\mathsans{K}$ and $\mathsans{K}_0$
on $P$ and $P_0$ with the expected properties.

\begin{defi} \label{defi:synthract}
The synthetic form of $\hat R$ is the map $R\in\Map(P\times\mathsans{K},P)$ given by 
\begin{align}
&R(V,A)=\hat R\circ(V\times A)
\vphantom{\Big]}
\label{ex2lgr2cm4}
\end{align}
for $V\in P$, $A\in\mathsans{K}$. 
The synthetic form of $\hat R_0$ is the map $R_0\in\Map(P_0\times\mathsans{K}_0,P_0)$  
yielded by restriction of $R$ to $P_0\times\mathsans{K}_0$. 
\end{defi}

\noindent
$R$ extends $\hat R$ because, writing $V\in P$ in the form \ceqref{ex2lgr2cm2} above, 
$R$ is given by $R(V,A)(\bar\alpha)=\hat R(\hat R(x,Q(\bar\alpha)),A(\bar\alpha))$ $=\hat R(x,Q(\bar\alpha)A(\bar\alpha))$
and thus has a range lying in $P$.  This calculation also shows that $R_0$, as the restriction of 
$R$ to $P_0\times\mathsans{K}_0$, has a range lying in $P_0$. 

\begin{prop} \label{prop:synthract}
$R$ and $R_0$ are right actions of $\mathsans{K}$ on $P$ and  $\mathsans{K}_0$ on $P_0$, respectively.  
Further, $\pi\circ R=\pi\circ\mathrm{pr}_1$ and $\pi_0\circ R_0=\pi_0\circ\mathrm{pr}_1$.
\end{prop}

\noindent
The above relations answer to \ceqref{2prinbund*/1}.  
 
\begin{proof} These properties follow trivially from the corresponding properties of $\hat R$ and $\hat R_0$ 
and $\hat\pi$ and $\hat\pi_0$. 
\end{proof}


For any trivializing neighborhood $U\subset M$ and trivializing functor $\hat\varPhi_U$, there exist trivializing maps 
$\varPhi_U$ and $\varPhi_{U0}$ of $P$ and $P_0$ induced by $\hat\varPhi_U$ and $\hat\varPhi_{U0}$. 

\begin{defi}\label{defi:phiusynth}
The synthetic form of $\hat\varPhi_U$ is the map $\varPhi_U\in\Map(\pi^{-1}(U),U\times\mathsans{K})$ 
given for $V\in\pi^{-1}(U)\subset P$ by 
\begin{align}
\varPhi_U(V)=\hat\varPhi_U\circ V.  
\vphantom{\Big]}
\label{ex2lgr2cm7}
\end{align}
The synthetic form of $\hat\varPhi_{U0}$ 
is the map $\varPhi_{U0}\in\Map(\pi_0{}^{-1}(U),U\times\mathsans{K}_0)$ 
yielded by restricting $\varPhi_U$ to $\pi_0{}^{-1}(U)$. 
\end{defi}

\noindent
$\varPhi_U$ and $\varPhi_{U0}$ having the ranges indicated follows from the following proposition. 

\begin{prop}\label{prop:phiusynth}
There is a $\mathsans{K}$--equivariant map $T_U\in\Map(\pi^{-1}(U),\mathsans{K})$ such that 
\begin{equation}
\varPhi_U(V)=(\pi(V),T_U(V)), \qquad T_U(V)=\hat T_U\circ V,
\label{ex2lgr2cm8}
\end{equation}
where the map $\hat T_U$ is defined through relation \ceqref{2prinbund14}. Further, $T_U$
restricts to a $\mathsans{K}_0$--equivariant map $T_{U0}\in\Map(\pi_0{}^{-1}(U),\mathsans{K}_0)$. 
\end{prop}

\noindent
\ceqref{ex2lgr2cm8} extends \ceqref{2prinbund14} and shows that $\varPhi_U$, as $\hat{\varPhi}_U$, 
is projection preserving. 

\begin{proof}
$T_U$ has range in $\mathsans{K}$, as, writing $V\in\pi^{-1}(U)$ in the form \ceqref{ex2lgr2cm2}, 
$T_U$ is given by $T_U(V)(\bar\alpha)=\hat T_U(V(\bar\alpha))=\hat T_U(\hat R(x,Q(\bar\alpha)))
=\hat T_{U0}(x)Q(\bar\alpha)$ with  $\hat T_{U0}(x)\in\hat{\mathsans{K}}_0$. 
\ceqref{ex2lgr2cm8} follows immediately from \ceqref{2prinbund14} and \ceqref{ex2lgr2cm3}.
The $\mathsans{K}$--equivariance of $T_U$ follows from the $\hat{\mathsans{K}}$--equivariance of $\hat T_U$.
The calculation above shows also that, for $v\in\pi_0{}^{-1}(U)$, $T_{U0}$ is given by 
$T_{U0}(v)(\bar\alpha)=\hat T_{U0}(x)$ so that $T_{U0}$ has range in $\mathsans{K}_0$. 
The $\mathsans{K}_0$--equivariance of $T_{U0}$ follows from the $\hat{\mathsans{K}}_0$--equivariance of $\hat T_{U0}$.
\end{proof}

\noindent
It is possible to define the synthetic analog $\tilde{\varPhi}_U$ of a weak inverse 
$\tilde{\hat{\varPhi}}_U$ of $\hat{\varPhi}_U$ and its object restriction. We shall not consider them here,
since we shall not need them in the following. 

In the rest of this section, of a mostly technical nature, we shall work out a graded differential geometric framework
for strict principal $2$--bundle theory based on the operational apparatus of subsect. \cref{subsec:opers} 
and the synthetic description expounded above. 
The operational synthetic formulation allows for an original theory of $2$--connections and $1$-- 
and $2$--gauge transformations, developed in depth in the companion paper II, very close in form and spirit to the 
principal bundle theoretic framework reviewed in subsect. \cref{subsec:scope}. 


\subsection{\textcolor{blue}{\sffamily Lie group and algebra crossed modules}}\label{subsec:liecm}

In subsect. \cref{subsec:synth}, we have seen that thanks to the isomorphism \ceqref{2prinbund3}
a strict Lie $2$--group is fully encoded in a pair of Lie groups equipped with two structure maps
with certain properties. These data 
constitute a Lie group crossed module, a notion which we review in this subsection mainly to set
our notation and terminology. We refer the reader to the papers 
\ccite{Baez5,Baez:2003fs} for a comprehensive treatment of this subject. 



\begin{defi}\label{defi:lgcm}
A Lie group crossed module $\mathsans{M}$ consists of two Lie groups $\mathsans{E}$ and $\mathsans{G}$ 
together with a Lie group morphism $\tau:\mathsans{E}\rightarrow\mathsans{G}$ and an action
$\mu:\mathsans{G}\times\mathsans{E}\rightarrow\mathsans{E}$ with the following properties.
\begin{enumerate}

\item For $a\in\mathsans{G}$, $\mu(a,\cdot)\in\Aut(\mathsans{E})$.

\item The map $a\in\mathsans{G}\rightarrow\mu(a,\cdot)\in\Aut(\mathsans{E})$
is a Lie group morphism.

\item Equivariance: for $a\in\mathsans{G}$, $A\in\mathsans{E}$
\begin{equation}
\tau(\mu(a,A))=a\tau(A)a^{-1}.
\label{liecm1}
\end{equation}

\item Peiffer identity: for $A,B\in\mathsans{E}$
\begin{equation}
\mu(\tau(A),B)=ABA^{-1}.
\label{liecm2}
\end{equation}

\end{enumerate}
\end{defi}

\noindent
In what follows, we shall often denote a Lie group crossed module $\mathsans{M}$ 
by the list of its constituent data $(\mathsans{E},\mathsans{G},\tau,\mu)$, 
or simply $(\mathsans{E},\mathsans{G})$ when no confusion can occur. 

\begin{defi}\label{defi:morlgcm}
A morphism $\beta:\mathsans{M}'\rightarrow\mathsans{M}$ of Lie group crossed modules
consists of two Lie group morphisms $\phi:\mathsans{G}'\rightarrow\mathsans{G}$  and 
$\varPhi:\mathsans{E}'\rightarrow\mathsans{E}$ with the following properties.
\begin{enumerate}

\item For $A\in\mathsans{E}'$ \hphantom{xxxxxxxxxxxx}
\begin{equation}
\tau(\varPhi(A))=\phi(\tau'(A)).
\label{liecm3}
\end{equation}

\item For $a\in\mathsans{G}'$, $A\in\mathsans{E}'$ \hphantom{xxxxxxxxxxxx}
\begin{equation}
\varPhi(\mu'(a,A))=\mu(\phi(a),\varPhi(A)).
\label{liecm4}
\end{equation}

\end{enumerate}
\end{defi}

\noindent
We shall concisely denote a Lie group crossed module morphism $\beta:\mathsans{M}'\rightarrow\mathsans{M}$ 
by the list $(\varPhi:\mathsans{E}'\rightarrow\mathsans{E},\phi:\mathsans{G}'\rightarrow\mathsans{G})$ 
of its defining data and shall omit the specification  
of sources and targets when no confusion is possible. 

Lie group crossed modules and morphisms thereof constitute a category denoted by $\bfs{\mathrm{Lgcm}}$. 

\begin{defi}\label{defi:lacm}
A Lie algebra crossed module $\mathfrak{m}$ consists of two Lie algebras $\mathfrak{e}$ and $\mathfrak{g}$
together with a Lie algebra morphism $t:\mathfrak{e}\rightarrow\mathfrak{g}$ and an action
$m:\mathfrak{g}\times\mathfrak{e}\rightarrow\mathfrak{e}$ with the following properties.  
\begin{enumerate}

\item For $u\in\mathfrak{g}$, $m(u,\cdot)\in\Der(\mathfrak{e})$.

\item The map $u\in\mathfrak{g}\rightarrow m(u,\cdot)\in\Der(\mathfrak{e})$
is a Lie algebra morphism.

\item Equivariance: for $u\in\mathfrak{g}$, $U\in\mathfrak{e}$
\begin{equation}
t(m(u,U))=[u,t(U)].
\label{liecm5}
\end{equation}

\item Peifer identity: for $U,V\in\mathfrak{e}$
\begin{equation}
m(t(U),V)=[U,V].
\label{liecm6}
\end{equation}

\end{enumerate}
\end{defi}

\noindent
As for a Lie group crossed module, we shall denote a Lie algebra crossed module $\mathfrak{m}$ 
by the list $(\mathfrak{e},\mathfrak{g},t,m)$ of its defining data or simply by $(\mathfrak{e},\mathfrak{g})$ 
when no confusion can arise. \vspace{1mm}\vfil\eject

\begin{defi}\label{defi:morlacm}
A morphism $p:\mathfrak{m}'\rightarrow\mathfrak{m}$ of Lie algebra crossed modules 
consists of two Lie algebra morphisms $h:\mathfrak{g}'\rightarrow\mathfrak{g}$ 
and $H:\mathfrak{e}'\rightarrow\mathfrak{e}$ with the following properties. 
\begin{enumerate}

\item For $U\in\mathfrak{e}'$ \hphantom{xxxxxxxxxxxx}
\begin{equation}
t(H(U))=h(t'(U)).
\label{liecm7}
\end{equation}

\item For $u\in\mathfrak{g}'$, $U\in\mathfrak{e}'$ \hphantom{xxxxxxxxxxxx}
\begin{equation}
H(m'(u,U))=m(h(u),H(U)).
\label{liecm8}
\end{equation}  

\end{enumerate}
\end{defi}

\noindent
We shall concisely denote a Lie algebra crossed module morphism $p:\mathfrak{m}'\rightarrow\mathfrak{m}$ 
by the list $(H:\mathfrak{e}'\rightarrow\mathfrak{e},h:\mathfrak{g}'\rightarrow\mathfrak{g})$ 
of its defining data and shall omit the specification  
of sources and targets when no confusion is possible. 

Lie algebra crossed modules and morphisms thereof constitute a category denoted by
$\bfs{\mathrm{Lacm}}$.  

\begin{prop}
With a Lie group crossed module $\mathsans{M}=(\mathsans{E},\mathsans{G},\tau,\mu)$, 
there is associated a Lie algebra crossed module 
$\mathfrak{m}=(\mathfrak{e},\mathfrak{g},t,m)$, where $\mathfrak{e}=\Lie\mathsans{E}$, 
$\mathfrak{g}=\Lie\mathsans{G}$, 
\begin{align}
&t=\dot\tau,
\vphantom{\Big]}
\label{liecm9}
\\
&m=\dot{}\mu{}\dot{}\,.
\vphantom{\Big]}
\label{liecm10}
\end{align}
\end{prop}

\noindent
Above, $\,\dot{}\,$ denotes Lie differentiation with respect to the relevant Lie group argument. 
In the second relation above, Lie differentiation is carried out with respect to both arguments. 
In the following, we shall often encounter the Lie differential $\mu{}\dot{}:\mathsans{G}\times
\mathfrak{e}\rightarrow\mathfrak{e}$. For $a\in\mathsans{G}$,
$\mu{}\dot{}\,(a,\cdot)\in\Aut(\mathfrak{e})$
and the map $a\in\mathsans{G}\rightarrow\mu{}\dot{}\,(a,\cdot)\in\Aut(\mathfrak{e})$
is a Lie group morphism. The Lie differential $\,\dot{}\mu:\mathfrak{g}\times
\mathsans{E}\rightarrow\mathfrak{e}$ shall also be considered.
See app. \cref{app:ident} for details about the precise definition and main properties 
of these objects. 

\begin{prop}
With every Lie group crossed module morphism $\beta:\mathsans{M}'\rightarrow\mathsans{M}$
$=(\varPhi:\mathsans{E}'\rightarrow\mathsans{E}$, $\phi:\mathsans{G}'\rightarrow\mathsans{G})$
there is associated by Lie differentiation
a Lie algebra crossed module morphism $p:\mathfrak{m}'\rightarrow\mathfrak{m}
=(H:\mathfrak{e}'\rightarrow\mathfrak{e},h:\mathfrak{g}'\rightarrow\mathfrak{g})$, 
where \pagebreak 
\begin{align}
&h=\dot\phi,
\vphantom{\Big]}
\label{liecm11}
\\
&H=\dot\varPhi.
\vphantom{\Big]}
\label{liecm12}
\end{align}
\end{prop}

\noindent
The morphism $p$ shall be denoted by $\dot\beta$ in the following.


The map that associates with each Lie group crossed module $\mathsans{M}$ its Lie algebra
crossed module $\mathfrak{m}$ and with each Lie group crossed module morphism 
$\beta:\mathsans{M}'\rightarrow\mathsans{M}$
its Lie algebra crossed module morphism $\dot\beta:\mathfrak{m}'\rightarrow\mathfrak{m}$ is a functor
of the category $\bfs{\mathrm{Lgcm}}$ into the category $\bfs{\mathrm{Lacm}}$.


\subsection{\textcolor{blue}{\sffamily Derived Lie groups and algebras}}\label{subsec:dercm}

The derived Lie group and algebra of a Lie group and algebra crossed module, respectively, 
constitute the fundamental algebraic structures of the formulation of $2$--connection and 
$1$-- and $2$--gauge transformation theory studied in depth in part II. 
In this subsection, we formulate these notions and study in detail their properties. 

Let $\mathsans{L}$ be an ordinary Lie group with Lie algebra $\mathfrak{l}$.
The internal function space $\iMap(\mathbb{R}[-1],\mathsans{L})$
with the Lie group structure inherited form that of $\mathsans{L}$ is a graded Lie group.
Every element $f\in\iMap(\mathbb{R}[-1],\mathsans{L})$
can be written uniquely as $f(\bar\alpha)=\ee^{\bar\alpha x}l$, $\bar\alpha\in\mathbb{R}[-1]$, 
with $l\in\mathsans{L}$ and $x\in\mathfrak{l}[1]$. 

Analogous remarks hold for an ordinary Lie algebra $\mathfrak{l}$. 
The internal function space $\iMap(\mathbb{R}[-1],\mathfrak{l})$
with the Lie algebra structure inherited form that of $\mathfrak{l}$ is a graded Lie algebra.
Any element $z\in\iMap(\mathbb{R}[-1],\mathfrak{l})$
can so be written uniquely as $z(\bar\alpha)=x+\bar\alpha y$, $\bar\alpha\in\mathbb{R}[-1]$, 
with $x\in\mathfrak{l}$ and $y\in\mathfrak{l}[1]$. 

For a Lie group crossed module $\mathsans{M}=(\mathsans{E},\mathsans{G},\tau,\mu)$,  
consider the graded Lie group $\iMap(\mathbb{R}[-1],\mathsans{E}\rtimes_\mu\mathsans{G})$, 
where $\mathsans{E}\rtimes_\mu\mathsans{G}$ denotes the semidirect product of the groups 
$\mathsans{E}$ and $\mathsans{G}$ with respect to the $\mathsans{G}$--action $\mu$. 
$\iMap(\mathbb{R}[-1],\mathsans{E}\rtimes_\mu\mathsans{G})$ contains among others elements of the special form 
\begin{equation}
P(\bar\alpha)=\ee^{\bar\alpha L}a, \quad\bar\alpha\in\mathbb{R}[-1],
\label{liecmx1}
\end{equation}
with $a\in\mathsans{G}$, $L\in\mathfrak{e}[1]$
forming form a distinguished subset $\DD\mathsans{M}$. 

For any $P,\,Q\in\DD\mathsans{M}$ \pagebreak with $P(\bar\alpha)=\ee^{\bar\alpha L}a$, $Q(\bar\alpha)=\ee^{\bar\alpha N}b$, one has 
\begin{align}
&PQ(\bar\alpha)=\ee^{\bar\alpha(L+\mu{}\dot{}(a,N))}ab,
\vphantom{\Big]}
\label{liecmx2}
\\
&P^{-1}(\bar\alpha)=\ee^{-\bar\alpha\mu{}\dot{}(a^{-1},L)}a^{-1}.
\vphantom{\Big]}
\label{liecmx3}
\end{align}
Therefore, $\DD\mathsans{M}$ is a graded Lie subgroup of $\iMap(\mathbb{R}[-1],\mathsans{E}\rtimes_\mu\mathsans{G})$. 

\begin{prop}
$\DD\mathsans{M}$ is a graded Lie group. The graded Lie group isomorphism 
\begin{equation}
\DD\mathsans{M}\simeq\mathfrak{e}[1]\rtimes_{\mu{}\dot{}}\mathsans{G} 
\label{liecm13}
\end{equation}
holds, where $\mathfrak{e}$ is regarded as an Abelian Lie group. 
\end{prop}

\noindent
Above, $\mathfrak{e}[1]\rtimes_{\mu{}\dot{}}\mathsans{G}$
denotes the semidirect product of the Lie groups $\mathfrak{e}[1]$ and $\mathsans{G}$ 
with respect to the $\mathsans{G}$--action $\mu{}\dot{}$. 
The graded Lie group $\DD\mathsans{M}$ is called the derived Lie group of $\mathsans{M}$. 

\begin{proof}
The statement follows immediately from  \ceqref{liecmx2}, \ceqref{liecmx3}. 
\end{proof}


\begin{prop}
With every Lie group crossed module morphism $\beta:\mathsans{M}'\rightarrow\mathsans{M}$
$=(\varPhi:\mathsans{E}'\rightarrow\mathsans{E}, \phi:\mathsans{G}'\rightarrow\mathsans{G})$ 
there is associated a graded Lie group morphism $\DD\beta:\DD\mathsans{M}'\rightarrow\DD\mathsans{M}$.
With respect to the factorization 
\ceqref{liecm13}, $\DD\beta$ reads 
\begin{equation}
\DD\beta=\dot\varPhi\times\phi.
\label{liecm13/1}
\end{equation}
\end{prop}

\noindent
$\DD\beta$ is called the derived Lie group morphism of $\beta$. 

\begin{proof}
The statement follows readily from  \ceqref{liecmx2}, \ceqref{liecmx3} and 
$\phi$, $\varPhi$ being group morphisms with $\varPhi$ satisfying 
\ceqref{liecm4}. 
\end{proof}


The map $\DD$ associating with each Lie group crossed module its derived Lie group and with each 
Lie group module morphism its derived Lie group morphism constitutes 
a functor from the Lie group crossed module category $\bfs{\mathrm{Lgcm}}$
into the graded Lie group category $\bfs{\mathrm{gLg}}$. 

Next, for a Lie algebra crossed module $\mathfrak{m}=(\mathfrak{e},\mathfrak{g},t,m)$, 
consider the internal function space $\iMap(\mathbb{R}[-1],\mathfrak{e}\rtimes_m\mathfrak{g})$, 
where $\mathfrak{e}\rtimes_m\mathfrak{g}$ denotes the semidirect product of the Lie algebras
$\mathfrak{e}$ and $\mathfrak{g}$ with respect to the $\mathfrak{g}$--action $m$. 
$\iMap(\mathbb{R}[-1],\mathfrak{e}\rtimes_m\mathfrak{g})$ contains in particular the elements of the form
\begin{equation}
Y(\bar\alpha)=u+\bar\alpha U, \quad\bar\alpha\in\mathbb{R}[-1],
\label{liecmx5}
\end{equation}
with $u\in\mathfrak{g}$, $U\in\mathfrak{e}[1]$
spanning a special subspace $\DD\mathfrak{m}$.

For any $Y,W\in\DD\mathfrak{m}$ such that $Y(\bar\alpha)=u+\bar\alpha U$, $W(\bar\alpha)=v+\bar\alpha V$, one has 
\begin{equation}
[Y,W](\bar\alpha)=[u,v]+\bar\alpha(m(u,V)-m(v,U)). 
\label{liecmx6}
\end{equation}
Thus, $\DD\mathfrak{m}$ is a graded Lie subalgebra of $\iMap(\mathbb{R}[-1],\mathfrak{e}\rtimes_m\mathfrak{g})$.

\begin{prop}
$\DD\mathfrak{m}$ is a graded Lie algebra. The graded Lie algebra isomorphism 
\begin{equation}
\DD\mathfrak{m}\simeq\mathfrak{e}[1]\rtimes_m\mathfrak{g}
\label{liecm14}
\end{equation}
holds, where 
$\mathfrak{e}$ is regarded as an Abelian Lie algebra. 
\end{prop}

\noindent
Above, $\mathfrak{e}[1]\rtimes_m\mathfrak{g}$
denotes the semidirect product of the Lie algebras
$\mathfrak{e}[1]$ and $\mathfrak{g}$ with respect to the $\mathfrak{g}$--action $m$.
The graded Lie algebra $\DD\mathfrak{m}$ is called the derived Lie algebra of $\mathfrak{m}$.

\begin{proof}
The claim is an immediate consequence of \ceqref{liecmx6}. 
\end{proof}



\begin{prop}
With every Lie algebra crossed module morphism $p:\mathfrak{m}'\rightarrow\mathfrak{m}$
$=(H:\mathfrak{e}'\rightarrow\mathfrak{e},h:\mathfrak{g}'\rightarrow\mathfrak{g})$
there is associated a graded Lie algebra morphism $\DD p:\DD\mathfrak{m}'\rightarrow\DD\mathfrak{m}$. 
With respect to the factorization \ceqref{liecm14}, $\DD p$ reads 
\begin{equation}
\DD p=H\times h.
\label{liecm14/1}
\end{equation}
\end{prop}

\noindent
$\DD p$ is called the derived Lie algebra morphism of $p$. 

\begin{proof}
The claim follows from \ceqref{liecmx6} and $H$, $h$ being Lie algebra morphisms with $H$
satisfying \ceqref{liecm8}. 
\end{proof}

The map $\DD$ associating with each Lie algebra crossed module its derived Lie algebra 
and with each Lie algebra crossed module morphism its derived Lie algebra morphism 
is a functor from the Lie algebra crossed module category 
$\bfs{\mathrm{Lacm}}$ into the graded Lie algebra category $\bfs{\mathrm{gLa}}$.  

Let $\mathsans{M}=(\mathsans{E},\mathsans{G},\tau,\mu)$ be a Lie group crossed module and let 
$\mathfrak{m}=(\mathfrak{e},\mathfrak{g},\dot\tau,{}\dot{}\mu{}\dot{})$ 
be its associated Lie algebra crossed module. 

\begin{prop}
$\DD\mathfrak{m}$ is the Lie algebra of $\DD\mathsans{M}$.
\end{prop}

\begin{proof}
A curve in $\DD\mathsans{M}$ is a smooth map $\varGamma:\mathbb{R}\rightarrow\DD\mathsans{M}$ 
such that $\varGamma(0)=1_{\mathsans{E}\rtimes_\mu\mathsans{G}}$. By \ceqref{liecmx1}, $\varGamma$
is of the form 
\begin{equation}
\varGamma(t)(\bar\alpha)=\ee^{\bar\alpha\varXi(t)}\gamma(t)
\label{}
\end{equation}
for smooth maps $\gamma:\mathbb{R}\rightarrow\mathsans{G}$ with $\gamma(0)=1_{\mathsans{G}}$ and 
$\varXi:\mathbb{R}\rightarrow\mathfrak{e}[1]$ with $\varXi(0)=0$. Differentiating $\varGamma$
at $t=0$, one finds 
\begin{equation}
Y(\bar\alpha):=\varGamma(t)^{-1}\frac{d\varGamma(t)}{dt}(\bar\alpha)\bigg|_{t=0}
=\gamma(t)^{-1}\frac{d\gamma(t)}{dt}\bigg|_{t=0}+\bar\alpha\frac{d\varXi(t)}{dt}\bigg|_{t=0}.
\label{}
\end{equation}
So, $Y\in\DD\mathfrak{m}$, as it is of the form \ceqref{liecmx5}. Moreover, every element $Y\in\DD\mathfrak{m}$ 
can be obtained in this way by a suitable choice of the maps $\gamma$ and $\varXi$.
\end{proof}

\noindent 
An explicit expression of the adjoint action of $\DD\mathsans{M}$ on $\DD\mathfrak{m}$
is available. 

\begin{prop}
For $P\in\DD\mathsans{M}$ and $Y\in\DD\mathfrak{m}$ respectively
of the form \ceqref{liecmx1} and \ceqref{liecmx5}, one has
\begin{align}
&\Ad P(Y)(\bar\alpha)
=\Ad a(u)+\bar\alpha(\mu{}\dot{}(a,U)-{}\dot{}\mu{}\dot{}(\Ad a(u),L)),
\vphantom{\Big]}
\label{liecmy5}
\\
&\Ad P^{-1}(Y)(\bar\alpha)=\Ad a^{-1}(u)+\bar\alpha\mu{}\dot{}(a^{-1},U+{}\dot{}\mu{}\dot{}(u,L)).
\vphantom{\Big]}
\label{liecmy6}
\end{align}
\end{prop}

\begin{proof} Let $D\in\DD\mathsans{M}$ be given by 
\begin{equation}
D(\bar\alpha)=\ee^{\bar\alpha U}\ee^u.
\label{liecmy7}
\end{equation}
Using the identities \ceqref{liecmx2}, \ceqref{liecmx3}, one finds
\begin{align}
&PDP^{-1}(\bar\alpha) 
=\ee^{\bar\alpha(\mu{}\dot{}(a,U)+L-\mu{}\dot{}(\ee^{\Ad a(u)},L))}\ee^{\Ad a(u)},
\vphantom{\Big]}
\label{liecmy8}
\\
&P^{-1}DP(\bar\alpha)    
=\ee^{\bar\alpha\mu{}\dot{}(a^{-1},\,U-L+\mu{}\dot{}(\ee^u,L))}\ee^{\Ad a^{-1}(u)}.
\vphantom{\Big]}
\label{liecmy9}
\end{align}
Linearizing these relations with respect to $u$, $U$, one obtains \ceqref{liecmy5}, \ceqref{liecmy6} immediately. 
\end{proof}

Let $\mathsans{M}$, $\mathsans{M}'$ be Lie group crossed modules and 
$\mathfrak{m}$, $\mathfrak{m}'$ be their associated Lie algebra crossed modules. 

\begin{prop}
If $\beta:\mathsans{M}'\rightarrow\mathsans{M}$ is a Lie group crossed module \pagebreak 
morphism and $\dot\beta:\mathfrak{m}'\rightarrow\mathfrak{m}$ 
is its associated Lie algebra crossed module morphism, then the Lie differential 
of the derived Lie group morphism $\DD\beta$ is the derived Lie algebra morphism $\DD\dot\beta$, that is 
\hphantom{xxxxxxxxxxxxxxxx}
\begin{equation}
\dot{\DD}\beta=\DD\dot\beta.
\label{liecmy10}
\end{equation}
\end{prop}

\begin{proof}
Let $\beta=(\varPhi,\phi)$. 
From \ceqref{liecm13/1}, it follows that $\dot{\DD}\beta=\dot\varPhi\times\dot\phi$.
From \ceqref{liecm14/1} with $p=\dot\beta$ and \ceqref{liecm11}, \ceqref{liecm12}, we have $\DD\dot\beta=\dot\varPhi\times\dot\phi$ too.
Identity \ceqref{liecmy10} so holds true. 
\end{proof}


For a Lie group crossed module $\mathsans{M}=(\mathsans{E},\mathsans{G},\tau,\mu)$, 
it is possible to construct appropriate Maurer--Cartan elements of the derived Lie group 
$\DD\mathsans{M}$ valued in the derived Lie algebra $\DD\mathfrak{m}$. 
Let $M$ denote a $\DD\mathsans{M}$ variable. 

\begin{defi}\label{defi:primmc}
The Maurer-Cartan forms of $\DD\mathsans{M}$ are the $1$--forms 
$dMM^{-1}$, $M^{-1}dM\in\Omega^1(\DD\mathsans{M})\otimes\DD\mathfrak{m}$, 
where $d$ is the de Rham differential of $\DD\mathsans{M}$.
\end{defi}

\noindent
Explicit formulae of the  Maurer-Cartan forms can be gotten by expressing $M$ as 
\begin{equation}
M(\bar\alpha)=\ee^{\bar\alpha E}g,\quad\bar\alpha\in\mathbb{R}[-1], 
\label{liecmy0}
\end{equation}
where $g$ and $E$ are a $\mathsans{G}$ and an $\mathfrak{e}[1]$ variable, respectively, 
in conformity with \ceqref{liecmx1},

\begin{prop}\label{prop:mcexpr}
The  Maurer-Cartan forms of $\DD\mathsans{M}$ are given by
\begin{align}
&dMM^{-1}(\bar\alpha)=dgg^{-1}-\bar\alpha(dE-\dot{}\mu\dot{}\,(dgg^{-1},E)),
\vphantom{\Big]}
\label{liecmy1}
\\
&M^{-1}dM(\bar\alpha)=g^{-1}dg-\bar\alpha\mu\dot{}\,(g^{-1},dE).
\vphantom{\Big]}
\label{liecmy2}
\end{align}
\end{prop}

\begin{proof}
\ceqref{liecmy1}, \ceqref{liecmy2}
follow straightforwardly from the well--known variational identities 
$\delta\ee^{\bar\alpha E}\ee^{-\bar\alpha E}=\frac{\exp(\bar\alpha\ad E)-1}{\bar\alpha\ad E}\delta(\bar\alpha E)$,
$\ee^{-\bar\alpha E}\delta\ee^{\bar\alpha E}=\frac{1-\exp(-\bar\alpha\ad E)}{\bar\alpha\ad E}\delta(\bar\alpha E)$ 
with $\delta=d$ upon taking the nilpotence of $\bar\alpha$ into account. 
\end{proof}


\subsection{\textcolor{blue}{\sffamily Derived Lie group and algebra valued internal functions}}\label{subsec:mapder}

In a graded geometric framework such as ours, 
$2$--connections and $1$--gauge transformations of a given principal $2$--bundle,
which we shall study in part II, are instances of derived Lie 
group and algebra valued internal functions. In this subsection, we 
study in detail the spaces of such type of functions. 

We begin with a few basic remarks. 
There exist a few definitions of infinite dimensional Lie groups and algebras and group 
and algebra morphisms and Lie differentiation, e. g.  Fr\'echet, diffeological etc.
In a typical infinite dimensional Lie theoretic setting, 
such notions coincide in their algebraic content but differ in their 
differential geometric one (unless the dimension is actually finite). For this reason, as long as 
only the algebraic structure is relevant and the smooth structure plays no role, 
we are free to leave the definition adopted unspecified. In such cases, we shall speak of {\it virtual} 
infinite dimensional Lie groups and algebras and Lie group and algebra morphisms and 
Lie differentiation and remain safely in the realm of pure algebra and formal analysis. 
Any statement of virtual infinite dimensional Lie theory turns into a statement of a specific kind  
of it provided the underlying geometric data satisfy suitable conditions.

We illustrate the above line of thought by reviewing very briefly the prototypical examples occurring in ordinary
infinite dimensional Lie theory. Let $M$ be a manifold. 
The space of functions from $M$ to a Lie group $\mathsans{L}$, $\Map(M,\mathsans{L})$, with the multiplication 
and inversion induced pointwise by those of $\mathsans{L}$ is a virtual infinite dimensional Lie group.  
Furthermore, any Lie group morphism $\lambda:\mathsans{L}'\rightarrow\mathsans{L}$ induces by left composition a
virtual infinite dimensional Lie group morphism $\Map(M,\lambda):\Map(M,\mathsans{L}')\rightarrow\Map(M,\mathsans{L})$.
Similarly, the space of functions from $M$ to a Lie algebra $\mathfrak{l}$, $\Map(M,\mathfrak{l})$, with 
the bracket induced pointwise by that of $\mathfrak{l}$ is a virtual infinite dimensional Lie algebra. 
Further, any Lie algebra morphism $l:\mathfrak{l}'\rightarrow\mathfrak{l}$ induces by left composition a virtual 
infinite dimensional Lie algebra morphism $\Map(M,l):\Map(M,\mathfrak{l}')\rightarrow\Map(M,\mathfrak{l})$.
When $\mathfrak{l}$ is the Lie algebra of $\mathsans{L}$, moreover, 
$\Map(M,\mathfrak{l})$ is the virtual Lie algebra of $\Map(M,\mathsans{L})$.  
The above virtual assertions 
become actual ones 
once a particular infinite dimensional Lie theoretic framework is adopted, 
provided $M$ meets suitable restrictions. 
For instance, in the Fr\'echet framework, those properties are verified if $M$ is compact. 

To the best of our knowledge, a suitable extension of the above framework 
to a graded differential geometric setting such as the one studied in this
subsection has not been worked out yet. Fortunately, we shall not need it. 
The purely algebraic notions of virtual infinite dimensional 
graded Lie group and algebra and group and algebra 
morphism, which are the obvious extensions to the graded setting 
of the corresponding notions of the ungraded one described above, will suffice. 

Below, we consider a fixed graded manifold $N$. The space of internal functions from $N$
to a graded Lie group $\mathsans{P}$, $\iMap(N,\mathsans{P})$, with the multiplication and inversion
induced pointwise by those of $\mathsans{P}$ is a virtual infinite dimensional 
graded Lie group. Furthermore, a virtual infinite dimensional 
graded Lie group morphism $\iMap(N,\kappa):\iMap(N,\mathsans{P}')\rightarrow\iMap(N,\mathsans{P})$
is determined by any graded Lie group morphism $\kappa:\mathsans{P}'\rightarrow\mathsans{P}$.
Similarly, the space of internal functions from $N$ 
to a graded Lie algebra $\mathfrak{p}$, $\iMap(N,\mathfrak{p})$, with the bracket induced pointwise by 
that of $\mathfrak{p}$ is a virtual infinite dimensional graded Lie algebra. 
A virtual infinite dimensional graded Lie algebra morphism 
$\iMap(N,k):\iMap(N,\mathfrak{p}')\rightarrow\iMap(N,\mathfrak{p})$
is furthermore determined by any graded Lie algebra morphisms $k:\mathfrak{p}'\rightarrow\mathfrak{p}$.
Similar statements hold also under full degree extension of the graded Lie algebras and algebra morphisms
involved (cf. app. \cref{subsec:convnot}). 
When  $\mathfrak{p}$ is the Lie algebra of $\mathsans{P}$, moreover, 
$\iMap(N,\mathfrak{p})$ is the virtual Lie algebra of $\iMap(N,\mathsans{P})$.  


Let $\mathsans{M}=(\mathsans{E},\mathsans{G},\tau,\mu)$ be a Lie group crossed module 
and let $\DD\mathsans{M}$ be its derived Lie group 
(cf. subsect. \cref{subsec:dercm}).  
The internal function space $\iMap(N,\DD\mathsans{M})$ is then a virtual infinite dimensional graded Lie group. 
$\iMap(N,\DD\mathsans{M})$ is most appropriately described by noticing that it is isomorphic to a subgroup 
of the virtual 
Lie group $\iMap(\mathbb{R}[1]\times N,\mathsans{E}\rtimes_\mu\mathsans{G})$. 
Indeed, by \ceqref{liecmx1}, an element $F\in\iMap(N,\DD\mathsans{M})$ reduces to a pair 
of internal functions $m\in\iMap(N,G)$, $Q\in\iMap(N,\mathfrak{e}[1])$ defining  the internal function 
$F\in\iMap(\mathbb{R}[1]\times N,\mathsans{E}\rtimes_\mu\mathsans{G})$, denoted here by the same symbol for simplicity, 
given by \hphantom{xxxxxxxxxxxxxxxx}
\begin{equation}
F(\alpha)=\ee^{\alpha Q}m, \quad \alpha\in\mathbb{R}[1].
\label{dercm5/1}
\end{equation}
Moreover, on account of \ceqref{liecmx2}, \ceqref{liecmx3}, for any
$F,G\in\iMap(N,\DD\mathsans{M})$ such that $F(\alpha)=\ee^{\alpha Q}m$, 
$G(\alpha)=\ee^{\alpha R}n$, we have 
\begin{align}
&FG(\alpha)=\ee^{\alpha(Q+\mu{}\dot{}(m,R))}mn,
\vphantom{\Big]}
\label{liecmx2/1}
\end{align} 
\begin{align}
&F^{-1}(\alpha)=\ee^{-\alpha\mu{}\dot{}(m^{-1},Q)}m^{-1}.
\vphantom{\Big]}
\label{liecmx3/1}
\end{align} 
The structure of the elements as well as the group multiplication and inversion of $\iMap(N,\DD\mathsans{M})$ is hence 
formally the same as that of $\DD\mathsans{M}$ given in eqs. \ceqref{liecmx1} and \ceqref{liecmx2}, \ceqref{liecmx3}
except for an inversion of degrees. 

\begin{prop}
The virtual infinite dimensional graded Lie group isomorphism  
\begin{equation}
\iMap(N,\DD\mathsans{M})\simeq\iMap(N,\mathfrak{e}[1])\rtimes_{\iMap(N,\mu{}\dot{})}\iMap(N,G)
\label{dercm5/2}
\end{equation}
holds true. 
\end{prop}

\noindent
\ceqref{dercm5/2} is the function space analog of \ceqref{liecm13}.

\begin{proof}
\ceqref{dercm5/2} follows from \ceqref{liecm13} by  acting with the functor $\iMap(N,-)$. 
\end{proof}

\begin{prop} With every Lie group crossed module morphism $\beta:\mathsans{M}'\rightarrow\mathsans{M}$
$=(\varPhi:\mathsans{E}'\rightarrow\mathsans{E},\phi:\mathsans{G}'\rightarrow\mathsans{G})$
there is associated a virtual infinite dimensional graded Lie group morphism 
$\iMap(N,\DD\beta):\iMap(N,\DD\mathsans{M}')\rightarrow\iMap(N,\DD\mathsans{M})$. With
respect to the factorization \ceqref{dercm5/2}, $\iMap(N,\DD\beta)$ reads as
\begin{equation}
\iMap(N,\DD\beta)=\iMap(N,\dot\varPhi)\times\iMap(N,\phi).  
\label{dercm5/3}
\end{equation}
\end{prop}

\noindent
This is the function space analog of \ceqref{liecm13/1}. 

\begin{proof}
\ceqref{dercm5/3} follows immediately from \ceqref{liecm13/1} through the action of the functor $\iMap(N,-)$ again. 
\end{proof}


Let $\mathfrak{m}=(\mathfrak{e},\mathfrak{g},t,m)$ be a Lie algebra crossed module 
and let $\DD\mathfrak{m}$ be its derived Lie algebra 
(cf. subsect. \cref{subsec:dercm}). The internal function space
$\iMap(N,\DD\mathfrak{m})$ is then a virtual infinite dimensional graded Lie algebra. 
$\iMap(N,\DD\mathfrak{m})$ is best described by observing that it is isomorphic to a subalgebra of the
virtual Lie algebra $\iMap(\mathbb{R}[1]\times N,\mathfrak{e}\rtimes_m\mathfrak{g})$. 
Indeed, by 
\ceqref{liecmx5}, an element $S\in\iMap(N,\DD\mathfrak{m})$ can be decomposed in 
a pair of internal functions $j\in\iMap(N,\mathfrak{g})$, $J\in\iMap(N,\mathfrak{e}[1])$, 
which together in turn define the internal function $S\in\iMap(\mathbb{R}[1]\times N,\mathfrak{e}\rtimes_m\mathfrak{g})$, 
denoted by the same symbol, given by \pagebreak 
\begin{equation}
S(\alpha)=j+\alpha J, \quad \alpha\in\mathbb{R}[1]. 
\label{dercm7/1}
\end{equation}
Further, on account of \ceqref{liecmx6}, for any 
$S,T\in\iMap(N,\DD\mathfrak{m})$ with $S(\alpha)=j+\alpha J$, 
$T(\alpha)=k+\alpha K$, we have 
\begin{equation}
[S,T](\alpha)=[j,k]+\alpha(m(j,K)-m(k,J)).
\label{liecmx6/1}
\end{equation}
Again, so, 
the structure of the elements and the Lie bracket of 
$\iMap(N,\DD\mathfrak{m})$ is formally the same as that of $\DD\mathfrak{m}$ given in eqs.  \ceqref{liecmx5} 
and \ceqref{liecmx6} except for an inversion of degrees. 

\begin{prop}
The virtual infinite dimensional graded Lie algebra isomorphism  
\begin{equation}
\iMap(N,\DD\mathfrak{m})\simeq\iMap(N,\mathfrak{e}[1])\rtimes_{\iMap(N,m)}\iMap(N,\mathfrak{g})
\label{dercm6/2}
\end{equation} 
holds true. 
\end{prop}

\noindent
\ceqref{dercm6/2} is the function space analog of \ceqref{liecm14}. 

\begin{proof}
\ceqref{dercm6/2} follows from \ceqref{liecm14} by acting with the functor $\iMap(N,-)$. 
\end{proof}

\begin{prop} 
With every Lie algebra crossed module morphism $p:\mathfrak{m}'\rightarrow\mathfrak{m}$
$=(H:\mathfrak{e}'\rightarrow\mathfrak{e},h:\mathfrak{g}'\rightarrow\mathfrak{g})$ 
there is associated a virtual infinite dimensional graded Lie algebra morphism 
$\iMap(N,\DD p):\iMap(N,\DD\mathfrak{m}')\rightarrow\iMap(N,\DD\mathfrak{m})$. 
With respect to the factorization \ceqref{dercm6/2}, $\iMap(N,\DD p)$ reads explicitly as 
\begin{equation}
\iMap(N,\DD p)=\iMap(N,H)\times\iMap(N,h).
\label{dercm6/3}
\end{equation} 
\end{prop}

\noindent
This is the function space analog of \ceqref{liecm14/1}.

\begin{proof}
Relation \ceqref{dercm6/3} follows readily from \ceqref{liecm14/1} through the action of the func\-tor $\iMap(N,-)$ again. 
\end{proof}

The derived Lie algebra $\DD\mathfrak{m}$ has a full degree prolongation 
\begin{equation}
\ZZ\DD\mathfrak{m}=\ZZ\mathfrak{g}\oplus\ZZ\mathfrak{e}=\ddd_{p=-\infty}^\infty\DD\mathfrak{m}[p],
\label{liecmx8}
\end{equation}
where \pagebreak $\DD\mathfrak{m}[p]=\mathfrak{e}[p+1]\oplus\mathfrak{g}[p]$ with 
$\DD\mathfrak{m}[0]=\DD\mathfrak{m}$ as a graded vector space (cf. eq. \ceqref{dercm2}). 
As $\DD\mathfrak{m}$ is a graded Lie algebra, the internal function space 
$\iMap(N,\ZZ\DD\mathfrak{m})$ is a virtual infinite dimensional 
graded Lie algebra containing $\iMap(N,\DD\mathfrak{m})$ as a Lie subalgebra.
$\iMap(N,\ZZ\DD\mathfrak{m})$ is moreover differential, as we now describe. To begin with, we note that 
$\iMap(N,\ZZ\DD\mathfrak{m})$ can be identified with a subalgebra of the 
virtual Lie algebra $\iMap(\mathbb{R}[1]\times N,\ZZ(\mathfrak{e}\rtimes_m\mathfrak{g}))$. 
Indeed, by \ceqref{liecmx8}, $\iMap(N,\ZZ\DD\mathfrak{m})$ decomposes as the direct sum 
\begin{equation}
\iMap(N,\ZZ\DD\mathfrak{m})=\ddd_{p=-\infty}^\infty\iMap(N,\DD\mathfrak{m}[p]).
\label{liecmx12/2}
\end{equation}
From \ceqref{liecmx5}, reasoning as we did earlier, an element $S\in\iMap(N,\DD\mathfrak{m}[p])$ 
comprises internal functions $j\in\iMap(N,\mathfrak{g}[p])$, $J\in\iMap(N,\mathfrak{e}[p+1])$ 
combining  in an internal function 
$S\in\iMap(\mathbb{R}[1]\times N,(\mathfrak{e}\rtimes_m\mathfrak{g})[p])$ given by 
\begin{equation}
S(\alpha)=j+(-1)^p\alpha J, \quad \alpha\in\mathbb{R}[1].
\label{dercm9/1}
\end{equation}
Further, on account of \ceqref{liecmx6}, 
one finds that the Lie bracket of a couple of elements $S\in\iMap(N,\DD\mathfrak{m}[p])$, 
$T\in\iMap(N,\DD\mathfrak{m}[q])$ such that $S(\alpha)=j+(-1)^p\alpha J$, 
$T(\alpha)=k+(-1)^q\alpha K$ is the element $[S,T]\in\iMap(N,\DD\mathfrak{m}[p+q])$ given by 
\begin{equation}
[S,T](\alpha)=[j,k]+(-1)^{p+q}\alpha(m(j,K)-(-1)^{pq}m(k,J)).
\label{liecmxx6/1}
\end{equation}
Identifying $\DD\mathfrak{m}[0]$ with $\DD\mathfrak{m}$ as usual, it is apparent here that 
$\iMap(N,\ZZ\DD\mathfrak{m})$ con\-tains $\iMap(N,\DD\mathfrak{m})$ as a virtual graded 
Lie subalgebra. In addition to a Lie algebra structure,  
$\iMap(N,\ZZ\DD\mathfrak{m})$ is also endowed with a cochain complex structure, 
\begin{align}
&\xymatrix@C=2pc
{\cdots\ar[r]^-{d_t}&\iMap(N,\DD\mathfrak{m}[p-1])\ar[r]^-{d_t}
&\iMap(N,\DD\mathfrak{m}[p])}
\label{liecmx11/1}
\\
&\hspace{6cm}\xymatrix@C=2pc{\ar[r]^-{d_t}
&\iMap(N,\DD\mathfrak{m}[p+1])\ar[r]^-{d_t}&\cdots},
\nonumber
\end{align} 
where the coboundary $d_t$ acts as 
\begin{equation}
d_tS(\alpha)=t\bigg(\frac{d}{d\alpha}S(\alpha)\bigg)=(-1)^pt(J)
\label{liecmx12/1}
\end{equation}
with $S\in\iMap(N,\DD\mathfrak{m}[p])$ of the form \ceqref{dercm9/1}. 

Under prolongation, the isomorphism in eq. \ceqref{dercm6/2} generalizes as stated in the next proposition.

\begin{prop} There is a virtual infinite dimensional 
differential graded Lie algebra isomorphism \hphantom{xxxxxxxxxxxx}
\begin{equation}
\iMap(N,\ZZ\DD\mathfrak{m})\simeq\iMap(N,\ZZ\mathfrak{e})\rtimes_{\iMap(N,s^{-1}t),\iMap(N,m)}\iMap(N,\ZZ\mathfrak{g}).
\label{dercm6/4}
\end{equation}
\end{prop}

\noindent
Above,  $\iMap(N,\ZZ\mathfrak{e})\rtimes_{\iMap(N,t),\iMap(N,m)}\iMap(N,\ZZ\mathfrak{g})$
denotes the differential semidirect product of the graded Lie algebras
$\iMap(N,\ZZ\mathfrak{e})$ and $\iMap(N,\ZZ\mathfrak{g})$, that is the ordinary Lie algebra 
semidirect product $\iMap(N,\ZZ\mathfrak{e})\rtimes_{\iMap(N,m)}\iMap(N,\ZZ\mathfrak{g})$
with the cochain complex structure 
$\iMap(N,s^{-1}t):\iMap(N,\ZZ\mathfrak{e})\rightarrow\iMap(N,\ZZ\mathfrak{g})$,
where $s^{-1}$ denotes the desuspension isomorphism lowering degree by one unit.

\begin{proof}
\ceqref{dercm6/4} follows by inspection of the structure of the Lie bracket $[-,-]$
shown in eq. \ceqref{liecmxx6/1} and the form of the coboundary $d_t$ shown in eq. \ceqref{liecmx12/1}.
and the fact that $d_t$ differentiates the Lie bracket \ceqref{liecmxx6/1} as it can be checked
by a straightforward calculation.
\end{proof}


\noindent
The $d_t$--cohomology of the complex is $\iMap(N,\ZZ\DD\mathfrak{h})$, 
where $\mathfrak{h}$ is the Lie algebra crossed module 
$(\ker t,\mathfrak{g}/\ran t)$ whose target morphism vanishes and whose $\mathfrak{g}/\ran t$--action 
is that induced by $m$. We shall not need this result however. 

\begin{prop}
Every Lie algebra crossed module morphism $p:\mathfrak{m}'\rightarrow\mathfrak{m}$ \linebreak 
$=(H:\mathfrak{e}'\rightarrow\mathfrak{e}$, $h:\mathfrak{g}'\rightarrow\mathfrak{g})$ induces
a virtual infinite dimensional differential 
graded Lie algebra morphism $\iMap(N,\ZZ\DD p):\iMap(N,\ZZ\DD\mathfrak{m}')
\rightarrow\iMap(N,\ZZ\DD\mathfrak{m})$. $\iMap(N,\ZZ\DD p)$ reads explicitly as 
\begin{equation}
\iMap(N,\ZZ\DD p)=\iMap(N,\ZZ H)\times\iMap(N,\ZZ h),
\label{dercm6/5}
\end{equation} 
\end {prop}

\noindent
where $\ZZ H$, $\ZZ h$ are the prolongations of $H$, $h$, respectively.

\begin{proof}
This follows from the factorization \ceqref{liecm14/1} 
under full degree extension for the semidirect product structure and 
and from relation \ceqref{liecm7} for the differential structure. 
\end{proof}



If $\mathfrak{m}$ is the Lie algebra crossed module of 
a Lie group crossed module $\mathsans{M}$, $\iMap(N,\DD\mathfrak{m})$
is the virtual Lie algebra of $\iMap(N,\DD\mathsans{M})$. 
The adjoint action of $\DD\mathsans{M}$ on $\DD\mathfrak{m}$ induces an 
adjoint action of $\iMap(N,\DD\mathsans{M})$ on $\iMap(N,\ZZ\DD\mathfrak{m})$.




\begin{prop}
For $F\in\iMap(N,\DD\mathsans{M})$ and $S\in\iMap(N,\DD\mathfrak{m}[p])$ 
respectively of the form \ceqref{dercm5/1} and \ceqref{dercm9/1}, one has 
\begin{align}
&\Ad F(S)(\alpha)
=\Ad m(j)+(-1)^p\alpha(\mu{}\dot{}(m,J)-{}\dot{}\mu{}\dot{}(\Ad m(j),Q)),
\vphantom{\Big]}
\label{liecmym11}
\\
&\Ad F^{-1}(S)(\alpha)=\Ad m^{-1}(j)+(-1)^p\alpha\mu{}\dot{}(m^{-1},J+{}\dot{}\mu{}\dot{}(j,Q)).
\vphantom{\Big]}
\label{liecmym12}
\end{align}
\end{prop}

\begin{proof} Relations \ceqref{liecmym11}, \ceqref{liecmym12} are an immediate consequence 
if identities \ceqref{liecmy5}, \ceqref{liecmy6}. 
\end{proof}

\begin{prop}
For a Lie group crossed 
module morphism $\beta:\mathsans{M}'\rightarrow\mathsans{M}$ with associated Lie algebra crossed 
module morphism $\dot\beta:\mathfrak{m}'\rightarrow\mathfrak{m}$, the virtual Lie algebra morphism
$\iMap(N,\DD\dot\beta)$ is the virtual Lie differential 
of the Lie group morphism $\iMap(N,\DD\beta)$, \hphantom{xxxxxxxxxxxx}
\begin{equation}
\dot{\iMap}(N,\DD\beta)=\iMap(N,\DD\dot\beta).
\label{liecmym13}
\end{equation}
\end{prop}

\noindent
This is the function space counterpart of relation \ceqref{liecmy10}. 

\begin{proof}
There really is no proof to be given here. In the spirit of the virtual Lie theory,
\ceqref{liecmym13} is essentially the definition of $\dot{\iMap}(N,\DD\beta)$. 
\end{proof}

Suppose that the manifold $N$ is equipped with a degree $p$ derivation $D$.
Given an internal function $C\in\iMap(N,\DD\mathsans{M})$, one can then construct 
elements $DCC^{-1}$, $C^{-1}DC\in\iMap(N,\DD\mathfrak{m}[p])$ by pull--back $C$ of the 
Maurer--Cartan forms $dMM^{-1}$, $M^{-1}dM$ of $\DD\mathsans{M}$ (cf. def. \cref{defi:primmc})
followed by contraction with $D$ seen as a vector field on $N$. Explicit formulae can be obtained by 
expressing $C$ as 
\begin{equation}
C(\alpha)=\ee^{\alpha O}r, \quad \alpha\in\mathbb{R}[1], 
\label{liecmyz0}
\end{equation}
with $r\in\iMap(N,\mathsans{G})$, $O\in\iMap(N,\mathfrak{e}[1])$, in conformity with \ceqref{dercm5/1},

\begin{prop}
$DCC^{-1}$, $C^{-1}DC$ are given explicitly by 
\begin{align}
&DCC(\alpha)^{-1}=Drr^{-1}+(-1)^p\alpha(DO-\dot{}\,\mu\dot{}\,(Drr^{-1},O)),
\vphantom{\Big]}
\label{liecmyz1}
\\
&C^{-1}DC(\alpha)=r^{-1}Dr+(-1)^p\alpha\mu\dot{}\,(r^{-1},DO). \hspace{2cm}
\vphantom{\Big]}
\label{liecmyz2}
\end{align} 
\end{prop}

\begin{proof}
Relations \ceqref{liecmyz1}, \ceqref{liecmyz2} \pagebreak are obtained proceeding as indicated above using the expressions of 
$dMM^{-1}$, $M^{-1}dM$  given in eqs. \ceqref{liecmy1}, \ceqref{liecmy2}. 
\end{proof}

\noindent
The coboundary $d_{\dot\tau}$ of $\iMap(N,\ZZ\DD\mathfrak{m})$ as a cochain complex 
allows for the construction of two special elements 
$d_{\dot\tau}CC^{-1}$, $C^{-1}d_{\dot\tau}C\in\iMap(N,\DD\mathfrak{m}[1])$. 
Explicit formulae can be obtained by 
expressing again $C$ as in \ceqref{liecmyz0}. 

\begin{prop}
$d_{\dot\tau}CC^{-1}$, $C^{-1}d_{\dot\tau}C$ are given by the expressions 
\begin{align}
&d_{\dot\tau}CC(\alpha)^{-1}=\dot\tau(O)+\frac{1}{2}\alpha[O,O],
\vphantom{\Big]}
\label{liecmyz3}
\\
&C^{-1}d_{\dot\tau}C(\alpha)=\Ad r^{-1}(\dot\tau(O))-\frac{1}{2}\alpha\mu\dot{}\,(r^{-1},[O,O]).
\vphantom{\Big]}
\label{liecmyz4}
\end{align}  
\end{prop}

\begin{proof}
\ceqref{liecmyz3}, \ceqref{liecmyz4} follow straightforwardly from the well--known variational identities 
$\delta\ee^{\alpha O}\ee^{-\alpha O}=\frac{\exp(\alpha\ad O)-1}{\alpha\ad O}\delta(\alpha O)$,
$\ee^{-\alpha O}\delta\ee^{\alpha O}=\frac{1-\exp(-\alpha\ad O)}{\alpha\ad O}\delta(\alpha O)$ 
with $\delta=d_{\dot\tau}=\dot\tau d/d\alpha$ and taking the nilpotence of $\alpha$ into account. 
\end{proof}


\subsection{\textcolor{blue}{\sffamily Derived Lie group and algebra cross modality}}\label{subsec:crossmdly}

In the setup of subsect. \cref{subsec:mapder}, 
one may consider the degenerate finite dimensional case where the graded manifold 
$N$ is the singleton manifold $*$. It is then natural to express everything 
in terms of the cross functor $(-)^+=\iMap(*,-)$ introduced in app. \cref{subsec:convnot}. 
This allows us to obtain a ``cross modality'' of the derived Lie group and algebras 
introduced in subsect. \cref{subsec:dercm}. 

In subsect. \cref{subsec:dercm}, it has been shown that 
with a given Lie group crossed module $\mathsans{M}=(\mathsans{E},\mathsans{G},\tau,\mu)$
there is associated a derived Lie group $\DD\mathsans{M}$. Through the cross functor $(-)^+$, 
the crossed modality derived Lie group $\DD\mathsans{M}^+$ can be also constructed.
$\DD\mathsans{M}^+$ enjoys properties analogous to those of $\DD\mathsans{M}$. 
The elements $F\in\DD\mathsans{M}^+$ are of the form \ceqref{dercm5/1} with  
$m\in\mathsans{G}^+\simeq\mathsans{G}$, $Q\in\mathfrak{e}[1]^+$ $\simeq\mathbb{R}[-1]\otimes\mathfrak{e}$
and the group operations read as in \ceqref{liecmx2/1}, \ceqref{liecmx3/1}, analogously to
\ceqref{liecmx1} and \ceqref{liecmx2}, \ceqref{liecmx3}. 
By \ceqref{dercm5/2}, we have further that 
$\DD\mathsans{M}^+\simeq\mathfrak{e}[1]^+\rtimes_{\mu{}\dot{}^+}\mathsans{G}^+$ as graded Lie groups,
similarly to \ceqref{liecm13}.
Finally, by virtue of \ceqref{dercm5/3}, for any Lie group crossed module morphism $\beta:\mathsans{M}'\rightarrow\mathsans{M}
=(\varPhi:\mathsans{E}'\rightarrow\mathsans{E}$, $\phi:\mathsans{G}'\rightarrow\mathsans{G})$
we have a graded Lie group morphism $\DD\beta^+:\DD\mathsans{M}^{\prime\,+}\rightarrow\DD\mathsans{M}^+$
with $\DD\beta^+=\dot\varPhi^+\times\phi^+$, analogously to \ceqref{liecm13/1}.  
The derived Lie group $\DD\mathsans{M}$ and its cross modality \pagebreak $\DD\mathsans{M}^+$ 
are however related in a deeper way as now we explain. 

The graded structure of 
$\DD\mathsans{M}$ stems from its being the subgroup of the internal function Lie group
$\iMap(\mathbb{R}[-1],\mathsans{E}\rtimes_{\mu}\mathsans{G})$ formed by the elements 
of the form \ceqref{liecmx1}. The degree $1$ carried by the variable $\bar\alpha\in\mathbb{R}[-1]$
is excluded in the degree counting. Upon including it instead, $\DD\mathsans{M}$ becomes a graded Lie group
concentrated in degree $0$, that can be treated as an ordinary Lie group. 
Completely analogous considerations apply for the cross modality derived Lie group $\DD\mathsans{M}^+$.
In the following, we shall regard these Lie groups in the way just 
described.

\begin{defi}\label{defi:crossmdly1}
Let $z_{\mathsans{M}}:\DD\mathsans{M}\rightarrow\DD\mathsans{M}^+$ be the map defined by 
\begin{equation}
z_{\mathsans{M}}P(\alpha)=\ee^{\alpha\,\zeta_{\mathfrak{e},1}(L)}a, \quad \alpha\in\mathbb{R}[1],
\label{crossmdly1}
\end{equation}
for $P\in\DD\mathsans{M}$ of the form \ceqref{liecmx1}.
\end{defi}

\noindent
Above, $\zeta_{\mathfrak{e},1}:\mathfrak{e}[1]\xrightarrow{~\simeq~}\mathfrak{e}[1]^+$ is the suspension 
isomorphism defined in app. \cref{subsec:convnot}. 

\begin{prop}\label{prop:crossmdly1}
$z_{\mathsans{M}}$ is a Lie group isomorphism.  
\end{prop}

\begin{proof}
Using \ceqref{liecmx2}, the linearity of $\zeta_{\mathfrak{e},1}$ and 
\ceqref{liecmx2/1}, one readily checks that $z_{\mathsans{M}}$ is a Lie group morphism.
The invertibility of $z_{\mathsans{M}}$ follows from that of $\zeta_{\mathfrak{e},1}$. 
\end{proof}   

In subsect. \cref{subsec:dercm}, it has been also shown that 
with a given Lie algebra crossed module $\mathfrak{m}=(\mathfrak{e},\mathfrak{g},t,m)$
there is associated a derived Lie algebra $\DD\mathfrak{m}$. Through the cross functor $(-)^+$, 
the crossed modality derived Lie algebra $\DD\mathfrak{m}^+$ can be also constructed.
$\DD\mathfrak{m}^+$ enjoys properties analogous to those of $\DD\mathfrak{m}$. 
The elements $S\in\DD\mathfrak{m}^+$ are of the form \ceqref{dercm7/1} with 
$j\in\mathfrak{g}^+\simeq\mathfrak{g}$, $J\in\mathfrak{e}[1]^+\simeq\mathbb{R}[-1]\otimes\mathfrak{e}$ 
and the Lie algebra operations read as in \ceqref{liecmx6/1}, similarly to \ceqref{liecmx5} and \ceqref{liecmx6}. 
By \ceqref{dercm6/2}, we have also that 
$\DD\mathfrak{m}^+\simeq\mathfrak{e}[1]^+\rtimes_{m^+}\mathfrak{g}^+$ as graded Lie algebras,
analogously to \ceqref{liecm14}.
Finally, by virtue of \ceqref{dercm6/3}, with every Lie algebra crossed module morphism 
$p:\mathfrak{m}'\rightarrow\mathfrak{m}$ 
$=(H:\mathfrak{e}'\rightarrow\mathfrak{e},h:\mathfrak{g}'\rightarrow\mathfrak{g})$ there is associated 
a graded Lie algebra morphism $\DD p^+:\DD\mathfrak{m}^{\prime\,+}\rightarrow\DD\mathfrak{m}^+$
with $\DD p^+=H^+\times h^+$, similarly to \ceqref{liecm14/1} Just as in the group case, the derived Lie algebra 
$\DD\mathfrak{m}$ and its cross modality $\DD\mathfrak{m}^+$ 
are related also in another deeper way. 

The graded structure of $\DD\mathfrak{m}$ stems \pagebreak from its being the subalgebra of the 
internal function Lie algebra $\iMap(\mathbb{R}[-1],\mathfrak{e}\rtimes_{m}\mathfrak{g})$ 
formed by the elements of the form \ceqref{liecmx5}. Including the degree $1$ of the variable
$\bar\alpha\in\mathbb{R}[-1]$ in the degree counting, $\DD\mathfrak{m}$ becomes a graded Lie algebra concentrated
in degree $0$ and as such can be treated as an ordinary Lie algebra. 
Completely analogous considerations apply for the cross mode derived Lie algebra 
$\DD\mathfrak{m}^+$. 
In the following, as in the group case, we shall regard these Lie algebras in the way just 
described.

\begin{defi}\label{defi:crossmdly2}
Let $\zeta_{\mathfrak{m}}:\DD\mathfrak{m}\rightarrow\DD\mathfrak{m}^+$ be the map defined by 
\begin{equation}
\zeta_{\mathfrak{m}}Y(\alpha)=u+\alpha\,\zeta_{\mathfrak{e},1}(U), \quad \alpha\in\mathbb{R}[1],
\label{crossmdly2}
\end{equation}
for $Y\in\DD\mathfrak{m}$ of the form \ceqref{liecmx5}.
\end{defi}

\begin{prop}\label{prop:crossmdly2}
$\zeta_{\mathfrak{m}}$ is a Lie algebra isomorphism.  
\end{prop}

\begin{proof}
Using \ceqref{liecmx6}, 
the linearity of $\zeta_{\mathfrak{e},1}$ and 
\ceqref{liecmx6/1}, one readily checks that $\zeta_{\mathfrak{m}}$ is a Lie algebra morphism.
The invertibility of $\zeta_{\mathfrak{m}}$ follows from that of $\zeta_{\mathfrak{e},1}$. 
\end{proof}

When $\mathfrak{m}$ is the Lie algebra crossed module of 
a Lie group crossed module $\mathsans{M}$, $\DD\mathfrak{m}^+$
is the Lie algebra of the Lie group $\DD\mathsans{M}^+$. Further, for a Lie group crossed 
module morphism $\beta:\mathsans{M}'\rightarrow\mathsans{M}$ with associated Lie algebra crossed 
module morphism $\dot\beta:\mathfrak{m}'\rightarrow\mathfrak{m}$, the Lie algebra morphism
$\DD \dot\beta^+$ is the Lie differential $\dot{\DD}\beta^+$ of the Lie group morphism $\DD\beta^+$.  

\begin{prop}
$\zeta_{\mathfrak{m}}$ is the Lie differential of $z_{\mathsans{M}}$.
\begin{equation}
\zeta_{\mathfrak{m}}=\dot z_{\mathsans{M}}.
\end{equation}
\end{prop}

\begin{proof}
This is apparent from inspecting \ceqref{crossmdly1}, \ceqref{crossmdly2}.
\end{proof}


\subsection{\textcolor{blue}{\sffamily Lie group crossed module spaces}}\label{subsec:cmsp}

\vspace{-.05mm}
The morphism and object manifolds of a principal $2$--bundle are equipped with the right action of the 
morphism and object groups of the structure Lie $2$--group, respectively. They are so special cases of 
Lie group spaces (cf. def. \cref{defi:lgrsp1}). 
In the synthetic crossed module theoretic version of the theory, Lie group spaces of this kind 
is instances of Lie group crossed module spaces, a notion we introduce next. 



\begin{defi}
A Lie group crossed module space $S$ is a Lie group space of the special form 
$(P,\DD\mathsans{M},R)$
for some graded manifold $P$, Lie group crossed module $\mathsans{M}$ 
and right action $R$ of $\DD\mathsans{M}$ on $P$.
\end{defi} 

\noindent 
Above, the derived Lie group $\DD\mathsans{M}$ (cf. subsect. \cref{subsec:dercm}) 
is regarded as an ordinary Lie group in the sense explained in subsect. \cref{subsec:crossmdly}. 
We shall denote the space $S$ through the list $(P,\mathsans{M},R)$ of its defining data.


A notion of Lie group crossed module space morphism can be formulated as a specialization of that of Lie
group space morphism. (cf. def. \cref{defi:lgrsp2}). 

\begin{defi}
A morphism $T:S'\rightarrow S$ of Lie group crossed module spaces consists in a morphism of the Lie group spaces 
$(P',\DD\mathsans{M}',R')$, $(P,\DD\mathsans{M},R)$ underlying $S'$, $S$ of the special form 
$(F:P'\rightarrow P$, $\DD\beta:\DD\mathsans{M}'\rightarrow\DD\mathsans{M})$ 
for some graded manifold morphism $F:P'\rightarrow P$ and 
Lie group crossed module morphism $\beta:\mathsans{M}'\rightarrow\mathsans{M}$.
\end{defi} 

\noindent 
Above, $\DD\beta$ is the derived Lie group morphism of $\beta$ 
(cf. subsect. \cref{subsec:dercm}). 
We shall denote the morphism $T:S'\rightarrow S$ 
by the list $(F:P'\rightarrow P,\beta:\mathsans{M}'\rightarrow\mathsans{M})$ of its defining data,
omitting indicating sources and targets when possible.

In this way, Lie group crossed module spaces and their morphisms form a category $\bfs{\mathrm{Lcmsp}}$
that can be identified with a subcategory of the category $\bfs{\mathrm{Lsp}}$ of Lie group spaces
defined in subsect. \cref{subsec:opers}. 

With any Lie group crossed module space there is associated an operation on general grounds
(cf. subsect. \cref{subsec:opers}, prop. \cref{prop:lgrsp1}).

\begin{defi}
For any Lie group  crossed module space $S=(P,\mathsans{M},R)$, 
we let $\iOOO S$ be the operation $(\iFun(T[1]P),\DD\mathfrak{m})$ of the Lie group space 
$(P,\DD\mathsans{M},R)$ underlying $S$. 
\end{defi}

\noindent
Above, $\mathfrak{m}$ is the Lie algebra crossed module associated with $\mathsans{M}$
(cf. subsect. \cref{subsec:dercm}). Further, $\DD\mathfrak{m}$ is regarded as an ordinary Lie algebra
in the sense explained in subsect. \cref{subsec:crossmdly}. 
In keeping with our notational conventions, we shall denote $\iOOO S$ 
concisely as $(\iFun(T[1]P),\mathfrak{m})$. 
We shall also denote the operation derivations by $d_P$, $j_P$, $l_P$. 
In similar fashion, any morphism of Lie group crossed module spaces determines 
a morphism of the associated operations again on general grounds
(cf. subsect. \cref{subsec:opers}, prop. \cref{prop:lgrsp2}).

\begin{defi}
For any Lie group crossed module space morphism $T:S'\rightarrow S$ $=
(F:P'\rightarrow P,\beta:\mathsans{M}'\rightarrow\mathsans{M})$, 
we denote by $\,\iOOO T:\iOOO S\rightarrow\iOOO S'$ the morphism 
$(F^*:\iFun(T[1]P)\rightarrow\iFun(T[1]P'), 
\DD\dot\beta:\DD\mathfrak{m}'\rightarrow\DD\mathfrak{m})$
of the operations $(\iFun(T[1]P),\DD\mathfrak{m})$, 
$(\iFun(T[1]P'),\DD\mathfrak{m}')$ underlying $\iOOO S$, $\iOOO S'$.
\end{defi}

\noindent
Above, $\dot\beta$ 
is the Lie algebra crossed module morphism  yield\-ed by the Lie group crossed module morphism 
$\beta$ 
by Lie differentiation and $\DD\dot\beta$ is the derived Lie
algebra morphism of $\dot\beta$ 
(cf. subsects. \cref{subsec:liecm}, \cref{subsec:dercm}). 

%

Consider  the special case of a morphism $T:S'\rightarrow S$ of Lie group crossed module spaces specified by a pair
$(F:P'\rightarrow P,\beta:\mathsans{M}'\rightarrow\mathsans{M})$, where $P=P'$ as graded 
manifolds and $F=\id_P$ as a graded manifold map. We then denote $S'$ by $\beta^*S$ and call it 
the pull-back of $S$ by $\beta$ and the operation $\iOOO S'$ by $\dot\beta^*\iOOO S$ and call it the 
pull-back by $\dot\beta$ of $\iOOO S$, so that $\iOOO\beta^*S=\dot\beta^*\iOOO S$. In fact, as Lie group spaces, 
$S'$ is just the pull--back $\DD\beta^*S$ of $S$ by the Lie group morphism $\DD\beta$. Furthermore, 
as operations of Lie group spaces, $\iOOO S'$ is just the pull--back $\DD\dot\beta{}^*\iOOO S$ of the operation $\iOOO S$
by the Lie algebra morphism $\DD\dot\beta$. See again  subsect. \cref{subsec:opers}. 

In the spirit of subsect. \cref{subsec:opers} above, we can think of a Lie group crossed module space 
$S=(P,\mathsans{M},R)$ as a generalized principal $\DD\mathsans{M}$--bundle
over $P/\DD\mathsans{M}$. In this way, we can identify the complexes
$(\iFun(T[1]P),d_P)$ and $(\iFun(T[1]P)_{\mathrm{basic}},d_P)$ of the associated operation $\iOOO S$ 
with the de Rham complexes $(\Omega^\bullet(P),d_{dR\,P})$ and $(\Omega^\bullet(P/\DD\mathsans{M}), 
d_{dR\,P/\DD\mathsans{M}})$, respectively.


\subsection{\textcolor{blue}{\sffamily Total space operations of  a principal 2--bundle}}\label{subsec:2prinbundop}

In this final subsection, we introduce and study the morphism and object space of a principal $2$--bundle 
and their associated operations. In this way, we have set the foundations for 
the operational total space theory of principal $2$--bundles, 
in particular of the $2$--connection and $1$-- and $2$-- gauge transformation theory of II, 
which is the goal of the present endeavour. 

Theor. \cref{theor:2lgr2cm} states that any strict Lie $2$--group $\hat{\matheul{K}}$ is fully 
described by a Lie group crossed module. 

\begin{prop}
The Lie group crossed module codifying $\hat{\matheul{K}}$ is $\mathsans{M}=(\mathsans{E},\mathsans{G},\tau,\mu)$,  
where $\mathsans{E}=\ker\hat{s}$, $\mathsans{G}=\hat{\mathsans{K}}_0$, $\tau=\hat{\epsilon}$
and $\mu=\hat{\lambda}$.
\end{prop}

\noindent Here, $\hat{s}$ 
is the source  map of $\hat{\mathsans{K}}$ and $\hat{\lambda}$ and $\hat{\epsilon}$ 
are the action and target maps defined in eqs. \ceqref{2prinbund4}
and \ceqref{2prinbund4/1}, respectively.  

\begin{proof}
This is just a restatement of theor. \cref{theor:2lgr2cm} in the notation of subsect. \cref{subsec:liecm}. 
\end{proof}

\begin{prop}
The Lie group crossed module $\mathsans{M}$ contains the submodule 
$\mathsans{M}_0=(1_{\mathsans{E}},\mathsans{G},\tau_0,\mu_0)$,  
where $\tau_0$ and $\mu_0$ are the restrictions of $\tau$ and $\mu$ to $1_{\mathsans{E}}\subset\mathsans{E}$
and $\mathsans{G}\times 1_{\mathsans{E}}\subset \mathsans{G}\times\mathsans{E}$, respectively. 
\end{prop}

\noindent
Note that $\tau_0$ and $\mu_0$ are necessarily trivial. 

\begin{proof}
The statement follows from $\tau_0$ and $\mu_0$ being restrictions of $\tau$ and $\mu$. 
\end{proof}

\noindent
In fact, $\mathsans{M}_0$ is just the discrete crossed module of the group $\mathsans{G}$. 
$\mathsans{M}_0$ provides a crossed module theoretic coding of $\mathsans{G}$,
which will turn out to be quite useful in the following. 

With the Lie Lie $2$--group $\hat{\matheul{K}}$, there are associated the synthetic Lie groups 
$\mathsans{K}$, $\mathsans{K}_0$ (cf. def. \cref{defi:synthk}). 
Recalling the derived construction expounded in subsect. \cref{subsec:dercm},
it is apparent that $\mathsans{K}$ is nothing but the 
derived Lie group $\DD\mathsans{M}$ of $\mathsans{M}$ 
\begin{equation}
\mathsans{K}=\DD\mathsans{M}.
\label{subsec:2prinbundop1}
\end{equation}
Analogously, $\mathsans{K}_0$ can be described as the derived group $\DD\mathsans{M}_0$, 
\begin{equation}
\mathsans{K}_0=\DD\mathsans{M}_0.
\label{subsec:2prinbundop2}
\end{equation}

In the synthetic formulation, with a principal $\hat{\matheul{K}}$--$2$--bundle $\hat{\mathcal{P}}$ 
there are associated synthetic forms of the morphism and object manifolds $\hat P$, $\hat P_0$, 
viz $P$, $P_0$ (cf. def. \cref{defi:synthp}). Importantly, $P$, $P_0$ are equipped with right 
$\mathsans{K}$--, $\mathsans{K}_0$--actions $R$, $R_0$, respectively (cf. def. \cref{defi:synthract} 
and prop. \cref{prop:synthract}). 

By \ceqref{subsec:2prinbundop1}, the $\mathsans{K}$--action $R$ of $P$ 
can be expressed as one of $\DD\mathsans{M}$. \pagebreak 
Relying on framework constructed in subsect. \cref{subsec:cmsp}, we can 
then state the following definition.

\begin{defi}
The morphism space of the principal $2$--bundle $P$ is the Lie group crossed module space 
$S_{P}=(P,\mathsans{M},R)$. 
\end{defi}

\noindent 
Similarly, by \ceqref{subsec:2prinbundop2}, the $\mathsans{K}_0$ --action $R_0$ of $P_0$ can be expressed as one 
of $\DD\mathsans{M}_0$. In the same framework, we can state a further definition.

\begin{defi}
The object space of the principal $2$--bundle $P$ is the Lie group crossed module space 
$S_{P0}=(P_0,\mathsans{M}_0,R_0)$. 
\end{defi}

\noindent
As detailed in subsect. \cref{subsec:cmsp}, 
with the spaces $S_{P}$ and $S_{P0}$ there are associated operations 
$\iOOO S_{P}=(\iFun(T[1]P),\mathfrak{m})$ and $\iOOO S_{P0}=(\iFun(T[1]P_0),\mathfrak{m}_0)$.

$S_{P0}$ and $\iOOO S_{P0}$ should be compatible with $S_{P}$ and $\iOOO S_{P}$ 
in the appropriate sense, as $P_0$ is a submanifold of $P$, $\mathsans{M}_0$ is a submodule of $\mathsans{M}$
and $R_0$ is a restriction of $R$. Such congruity is codified by an appropriate crossed module space morphism
and its appended operation morphism along the lines of subsect. \cref{subsec:cmsp}. 

\begin{defi}
Let $L:S_{P0}\rightarrow S_{P}$ be the Lie group crossed module space morphism 
$L=(I,\varsigma)$, where $I:P_0\rightarrow P$

and $\varsigma:\mathsans{M}_0\rightarrow\mathsans{M}$ are the inclusion maps of $P_0$ and
$\mathsans{M}_0$ into $P$ and $\mathsans{M}$, respectively. 
\end{defi} 

\noindent 
Associated with $L$ there is a  morphism $\iOOO L:\iOOO S_{P}\rightarrow\iOOO S_{P0}$ of operations, 
namely $\iOOO L=(I^*,\dot\varsigma)$, $I^*:\iFun(P)\rightarrow\iFun(P_0)$ being the restriction pull-back 
and $\dot\varsigma:\mathfrak{m}_0\rightarrow\mathfrak{m}$ the Lie algebra crossed module morphism induced by the 
$\varsigma$ by Lie differentiation. 

Our operational synthetic theory so closely parallels that of the graded geometric version of 
the ordinary theory reviewed in subsect. \cref{subsec:scope} with a few important differences. 
To begin with, unlike the ordinary theory, it is not directly based on the relevant principal 
$\hat{\matheul{K}}$--$2$--bundle $\hat{\mathcal{P}}$ but on its attached synthetic setup. This 
hinges on the synthetic morphism and object groups $\mathsans{K}$, $\mathsans{K}_0$
and manifolds $P$, $P_0$, which, as discussed at length in subsect. \cref{subsec:synth},
do not constitute a true synthetic $\matheul{K}$--$2$--bundle $\mathcal{P}$ in spite 
of having many properties of one. Moreover, again unlike the ordinary theory, 
two operations, viz $\iOOO S_{P}$, $\iOOO S_{P0}$,  rather that just one appear and are potentially relevant. 
More technically, furthermore, the internal function algebras $\iFun(T[1]P)$, $\iFun(T[1]P_0)$ instead 
of the ordinary algebras $\Fun(T[1]P)$, $\Fun(T[1]P_0)$ are used here, 
a feature that the end is bound to make a difference at the end.  

The above dissimilarities notwithstanding, the ordinary theory provides a simple model, reference to which 
considerably aids intuition. In particular, it is useful to think of the whole synthetic setup of 
$\mathsans{K}$, $\mathsans{K}_0$ and $P$, $P_0$ as if it were some kind of 
synthetic $\matheul{K}$--$2$--bundle $\mathcal{P}$, though, as we have recalled,  
strictly speaking it is not lacking as it does a compositional structure.

\appendix

\vfil\eject

\section{\textcolor{blue}{\sffamily Notation and conventions}}\label{subsec:convnot}

We recall below some of the basic notions and conventions of graded algebra and geometry 
we use throughout the present paper. 

For a pair $M$, $N$ of graded manifolds, we denote by $\iMap(M,N)$ 
the set of internal functions of $M$ into $N$. Thus, when expressed in terms of local 
body and soul coordinates $t^a$ and $y^r$ of $M$, the components of one such function 
with respect to local body and soul coordinates $u^i$ and $z^h$ of $N$ are polynomials
in the $y^r$ with coefficients which are smooth functions of the $t^a$ of possibly non zero degree. 
When $N$ is a graded vector space, group, Lie algebra etc.,
so is $\iMap(M,N)$ with the pointwise operations induced by those of $N$. 

In what follows, $E$ stands for an ungraded finite dimensional real vector space. 
Much of what we shall say can be formulated also for a graded vector space, but we shall not 
need to do so. 

For any integer $p$, we denote by $E[p]$  the degree $p$ shift of $E$, 
a copy of $E$ placed in degree $-p$. If we conventionally think of $E$ 
as a graded vector space supported in degree $0$, as it is customarily done, we may 
identify $E[0]$ with $E$ itself. 

A linear coordinate of $E[p]$ is just a non zero element of $E[p]^\vee=E^\vee[-p]$ and has 
hence degree $p$, where ${}^\vee$ denotes duality. Given a set of vectors $e_i\in E$ 
constituting a basis, there exists a 
unique set of linear coordinates $x_p{}^i$ of $E[p]$ dual to the basis, that is such that 
$x_p{}^i(e_j)=s^p\delta^i{}_j$, where $s^p$ denotes $p$--fold suspension raising degree by $p$ units. 

If we equip $E[p]$ with a set of linear coordinates $x_p{}^i$ as above, $E[p]$ becomes a graded manifold with 
singleton body. As such, $E[p]$ is concentrated in degree $p$ because the $x_p{}^i$ have degree $p$. 
For a graded manifold $M$, the set $\iMap(M,E[p])$ of $E[p]$--valued internal functions of $M$ has so a natural 
structure of graded vector space supported in degree $p$. 

The particular case where $E=\mathbb{R}$ deserves a special mention for its relevance. 
For each $p$, we have the graded vector space 
\begin{equation}
\iFun_p(M)=\iMap(M,\mathbb{R}[p])
\label{convnot2}
\end{equation}
of $\mathbb{R}[p]$--valued internal functions. The spaces $\iFun_p(M)$ for varying $p$ span together the 
graded vector space \hphantom{xxxxxxxxxxxx}
\begin{equation}
\iFun(M)=\ddd_{p=-\infty}^\infty\iFun_p(M). 
\label{convnot1}
\end{equation}
Thanks to the existence of graded commutative products 
$\iFun_p(M)\times\iFun_q(M)$ $\rightarrow\iFun_{p+q}(M)$,
$\iFun(M)$ is a graded commutative algebra. $\iFun(M)$ is just the algebra of internal functions of $M$ 
and for each $p$ $\iFun_p(M)$ is the subspace of $\iFun(M)$ of degree $p$ internal functions. 

There exists a canonical graded vector space isomorphism 
\begin{equation}
\iMap(M,E[p])\simeq\iFun_p(M)\otimes E.
\label{dercm3}
\end{equation}  
Indeed, upon choosing basis vectors $e_i$ of $E$, a given function $f\in\iMap(M,E[p])$ 
is fully specified by a set of degree $p$ internal functions $f^i\in\iFun_p(M)$, 
the components of $f$ with respect to the linear coordinates $x_p{}^i$ dual to the $e_i$. 
The $f^i$ in turn define an element $\hat f\in\iFun_p(M)\otimes E$ given by 
\begin{equation}
\hat f=f^i\otimes e_i. 
\label{dercm3/1}
\end{equation}
$\hat f$ is by construction independent from the choice of the basis $e_i$. 
The correspondence $f\rightarrow\hat f$ yields the isomorphism \ceqref{dercm3}. 
By \ceqref{dercm3}, 
the functions of $\iMap(M,E[p])$ can be regarded as specially structured collections of functions of 
$\iFun_p(M)$. In this way, further, any linear operation on $\iFun(M)$, e. g. a derivation, immediately 
induces a corresponding operation on $\iMap(M,E[p])$. 

The function space $\iMap(E[p],E[p])$ contains a tautological element $x_p$, 
\begin{equation}
x_p=x_p{}^i\otimes e_i,
\label{dercm3/3}
\end{equation}
of degree $p$. $x_p$ corresponds to the identity function $\id_{E[p]}$ of $E[p]$ regarded as a graded manifold.

The internal function spaces $\iMap(*,N)$ with $*$ the singleton manifold play an important role in 
the analysis of this paper. $\iMap(*,-)$ is a functor from the category of graded manifolds to the category 
of sets. We call it the cross functor and denote it with the simplified notation $(-)^+$. 

When restricted to graded vector spaces regarded as graded manifolds, the cross functor 
can be described rather explicitly. For $\mathbb{R}$, by \ceqref{convnot2}, we have 
\begin{equation}
\mathbb{R}[p]^+
=\iFun_p(*)\simeq \mathbb{R}[-p]. 
\label{convnot3}
\end{equation}
By \ceqref{dercm3}, for the vector space $E$, we have similarly  \hphantom{xxxxxxxxxxxx} 
\begin{equation}
E[p]^+\simeq\iFun_p(*)\otimes E 
\simeq \mathbb{R}[-p]\otimes E. 
\label{convnot4}
\end{equation}
Thus, by \ceqref{dercm3/1}, chosen a basis $e_i$ of $E$ an element $v\in E[p]^+$ expands as %
\begin{equation}
v=v^i\otimes e_i, 
\label{dercm3/2}
\end{equation}
where $v^i\in\mathbb{R}[-p]$. 
Notice that $\mathbb{R}[p]^+$ differs from $\mathbb{R}[p]$ for $p\not=0$, being in fact 
$\mathbb{R}[p]^+=s^{2p}\mathbb{R}[p]$. 
Similarly, $E[p]^+$ differs from $E[p]$ but it is related to $E[p]$ be means of the $2p$--fold
suspension isomorphism $\zeta_{E,p}:E[p]\xrightarrow{~\simeq~}E[p]^+$
defined by the sequence 
$E[p]\simeq\mathbb{R}\otimes E[p]\xrightarrow{~s^p\otimes s^p~}
s^p\mathbb{R}\otimes s^pE[p]\simeq\mathbb{R}[-p]\otimes E\simeq E[p]^+$. 

It is sometimes useful to assemble the degree shifts $E[p]$ of $E$ for all $p$'s 
into the full degree extension $\ZZ E$ of $E$, the infinite dimensional graded vector space
\begin{equation}
\ZZ E=\ddd_{p=-\infty}^\infty E[p]. 
\label{dercm2}
\end{equation}
Similarly, we can assemble the $E[p]$--valued internal functions for  all $p$'s 
into the graded vector space
\begin{equation}
\iMap(M,\ZZ E)=\ddd_{p=-\infty}^\infty \iMap(M,E[p]). 
\label{dercm2/1}
\end{equation}
When $E$ is endowed with additional structures such as those of graded commutative algebra,
graded Lie algebra, etc, these are inherited by $\iMap(M,\ZZ E)$. 

An important instance is the full degree extension $\ZZ \mathbb{R}$ of 
$\mathbb{R}$. $\iMap(M,\ZZ\mathbb{R})$ is a graded commutative algebra 
thanks to the ordinary multiplicative structure of $\mathbb{R}$. Indeed, 
$\iMap(M,\ZZ\mathbb{R})$ is nothing but the internal function algebra 
$\iFun(M)$ described earlier. 

In this paper, we are mostly concerned with internal function spaces of graded manifolds
such as $\iMap(M,N)$. The above setup can also be formulated for ordinary function spaces
$\Map(M,N)$ with similar results but also a few noticeable exceptions. \pagebreak Such spaces are more restrictive 
than the internal ones, since the coefficient functions of soul coordinates polynomials are 
required to be degree $0$ functions of the body coordinates. For the singleton manifold $M=*$, one has for instance
$\Map(*,\mathbb{R}[p])=\delta_{p,0}\mathbb{R}$ and $\Map(*,E[p])=\delta_{p,0}E$.




\section{\textcolor{blue}{\sffamily Basic identities}}\label{app:ident}

Let $\mathsans{M}=(\mathsans{E},\mathsans{G},\tau,\mu)$ be a Lie group crossed module.
We collect below a number of structure relations which are used throughout the main text of the paper.  

The relevant differentiated structure maps are $\dot\tau:\mathfrak{e}\rightarrow\mathfrak{g}$,
$\mu\dot{}\,:\mathsans{G}\times\mathfrak{e}\rightarrow\mathfrak{e}$,
$\dot{}\mu:\mathfrak{g}\times\mathsans{E}\rightarrow\mathfrak{e}$ and 
$\dot{}\mu\dot{}\,:\mathfrak{g}\times\mathfrak{e}\rightarrow\mathfrak{e}$. They are defined as 
\begin{align}
&\dot\tau(X)=\frac{d\tau(C(v))}{dv}\Big|_{v=0},
\vphantom{\Big]}
\label{ident1}
\\
&\mu\dot{}\,(a,X)=\frac{d}{dv}\mu(a,C(v))\Big|_{v=0},
\vphantom{\Big]}
\label{ident2}
\\
&\dot{}\mu(x,A)=\frac{d}{du}\mu(c(u),A)A^{-1}\Big|_{u=0}, 
\vphantom{\Big]}
\label{ident3}
\\
&\dot{}\mu\dot{}\,(x,X)
=\frac{\partial}{\partial u}\Big(\frac{\partial\mu(c(u),C(v))}{\partial v}\Big|_{v=0}\Big)\Big|_{u=0}
\vphantom{\Big]}
\label{ident4}
\end{align}
for $a\in\mathsans{G}$, $A\in\mathsans{E}$,
$x\in\mathfrak{g}$, $X\in\mathfrak{e}$, where $c(u)$ and $C(v)$ are curves in $\mathsans{G}$ 
and $\mathsans{E}$ with  $c(u)\big|_{u=0}$ $=1_{\mathsans{G}}$ and $C(v)\big|_{v=0}=1_{\mathsans{E}}$ and 
$dc(u)/du\big|_{u=0}=x$ and $dC(v)/dv\big|_{v=0}$ $=X$, respectively, whose choice is immaterial. 

The following algebraic identities hold:
\begin{align}
&\dot\tau(\,\,\dot{}\mu(x,A))=x-\Ad\tau(A)(x), 
\vphantom{\Big]}
\label{ident5}
\\
&\,\dot{}\mu(\dot\tau(X),A)=X-\Ad A(X), 
\vphantom{\Big]}
\label{ident6}
\\
&\,\dot{}\mu([x,y],A)
=\dot{}\mu\dot{}\,(x,\,\dot{}\mu(y,A))-\dot{}\mu\dot{}\,(y,\,\dot{}\mu(x,A))-[\,\,\dot{}\mu(x,A),\,\dot{}\mu(y,A)],
\vphantom{\Big]}
\label{ident7}
\\
&\,\dot{}\mu(x,AB)=\dot{}\mu(x,A)+\Ad A(\,\,\dot{}\mu(x,B)), 
\vphantom{\Big]}
\label{ident8}
\\
&\,\dot{}\mu(\Ad a(x),\mu(a,A))=\mu\dot{}\,(a,\,\dot{}\mu(x,A)),
\vphantom{\Big]}
\label{ident9}
\\
&\Ad A(\,\dot{}\mu\dot{}\,(x,X))=\dot{}\mu\dot{}\,(x,\Ad A(X))-[\,\,\dot{}\mu(x,A),\Ad A(X)],
\vphantom{\Big]}
\label{ident110}
\end{align}
where $a\in\mathsans{G}$, $A,B\in\mathsans{E}$, $x,y\in\mathfrak{g}$, $X\in\mathfrak{e}$.

The following variational identities hold: \pagebreak
\begin{align}
&\delta\mu(a,A)\mu(a,A)^{-1}
=\mu\dot{}\,(a,\,\dot{}\mu(a^{-1}\delta a,A)+\delta AA^{-1}), 
\vphantom{\Big]}
\label{ident11}
\\
&\delta\mu\dot{}\,(a,X)=\mu\dot{}\,(a,\dot{}\mu\dot{}\,(a^{-1}\delta a,X)+\delta X), 
\vphantom{\Big]}
\label{ident12}
\\
&\delta\,\,\dot{}\mu(x,A)=\dot{}\mu(\delta x,A)+\dot{}\mu\dot{}\,(x,\delta AA^{-1})-[\,\,\dot{}\mu(x,A),\delta AA^{-1}],
\vphantom{\Big]}
\label{ident13}
\end{align}
where $a\in\mathsans{G}$, $A\in\mathsans{E}$, $x\in\mathfrak{g}$, $X\in\mathfrak{e}$.

\vfil\eject

\noindent
\textcolor{blue}{Acknowledgements.} 
The author thanks R. Picken, J. Huerta and C. Saemann for useful discussions.
He acknowledges financial support from INFN Research Agency
under the provisions of the agreement between University of Bologna and INFN. 
He also thanks the organizer of the 2018 EPSRC Durham Symposium 
on ``Higher Structures in M-Theory'' during which part of this work was done.

\vfil\eject

\end{document}